\documentclass[12pt,a4paper]{article}

\usepackage[margin=0.75in]{geometry} 

\usepackage{amsmath, amsthm, amssymb, amsfonts}
\usepackage{booktabs} 
\usepackage{makecell} 
\usepackage{dsfont}
\usepackage{algorithm}
\usepackage{algpseudocode}

\usepackage[colorlinks=true,allcolors=black]{hyperref}
\usepackage{tikz}
\usetikzlibrary{arrows,automata}

\usepackage[backend=bibtex,style=numeric,sorting=ynt]{biblatex}

\title{On the optimality of Shapley value mechanism for funding public excludable goods under Sybil strategies}
\author{Bruno Mazorra \\ Universitat Pompeu Fabra \\ brunomazorra@gmail.com}

\newcommand{\1}{\textbf{1}^T}

\usepackage{pgfplots}
\usepackage{fancyhdr}
\usepackage{tikz-cd}
\usepackage{graphicx}
\usepackage{bm}
\usepackage[skins,breakable]{tcolorbox}
\graphicspath{ {./images/} }
\usepackage{float}

\newtcolorbox{mybox}[1][]{enhanced jigsaw,breakable,pad at break=1mm,
  oversize,left=8mm,interior hidden,colframe=black,nobeforeafter=,#1}

\newtcolorbox{mybox2}[2][]{%
  attach boxed title to top center
               = {yshift=-8pt},
  colback      = pink!5!white,  
  colframe     = pink!30!black, 
  fonttitle    = \bfseries,
  colbacktitle = pink!50!black, 
  title        = #2,#1,
  enhanced,
}
\newtheoremstyle{normalfonttheorem} 
  {} 
  {} 
  {\normalfont} 
  {} 
  {\bfseries} 
  {.} 
  {5pt plus 1pt minus 1pt} 
  {} 

\theoremstyle{normalfonttheorem}

\newtheorem{theorem}{Theorem}[section]
\newtheorem{lemma}[theorem]{Lemma}
\newtheorem{corollary}[theorem]{Corollary}
\newtheorem{proposition}[theorem]{Proposition}
\newtheorem{definition}[theorem]{Definition}
\newtheorem{observation}[theorem]{Observation}

\newcommand{\x}{\textbf{x}}
\newcommand{\p}{\textbf{p}}
\newcommand{\bb}{\textbf{b}}
\newcommand{\Rr}{\mathbb{R}}
\newcommand{\keyword}[1]{\hspace{0.5cm}\textbf{Keywords}: #1}

\renewcommand{\sectionmark}[1]{}
\renewcommand{\subsectionmark}[1]{}

\begin{document}
\begin{titlepage}
    
\pagestyle{fancy}
\fancyhf{} 
\maketitle
\begin{abstract}
In the realm of cost-sharing mechanisms, the vulnerability to Sybil strategies —also known as false-name strategies, where agents create fake identities to manipulate outcomes— has not yet been studied. In this paper, we delve into the details of different cost-sharing mechanisms proposed in the literature, highlighting their non-Sybil-resistant nature. Furthermore, we prove no deterministic, anonymous, truthful, Sybil-proof, upper semicontinuous, and individually rational cost-sharing mechanism for public excludable goods is better than $\Omega(n)-$approximate. This finding reveals an exponential increase in the worst-case social cost in environments where agents are restricted from using Sybil strategies. To circumvent these negative results, we introduce the concept of \textit{Sybil Welfare Invariant} mechanisms, where a mechanism does not decrease its welfare under Sybil strategies when agents choose weak dominant strategies and have subjective prior beliefs over other players' actions. Finally, we prove that the Shapley value mechanism for symmetric and submodular cost functions holds this property, and so deduce that the worst-case social cost of this mechanism is the $n$th harmonic number $\mathcal H_n$ under equilibrium with Sybil strategies, matching the worst-case social cost bound for cost-sharing mechanisms. This finding suggests that any group of agents, each with private valuations, can fund public excludable goods both permissionless and anonymously, achieving efficiency comparable to that of non-anonymous domains, even when the total number of participants is unknown.
\end{abstract}
\keyword{Permisionless Mechanism design, Cost sharing, Sybil-proof}
\newpage
\tableofcontents
\end{titlepage}
\setcounter{page}{1}
\pagenumbering{arabic}
\section{Introduction}

Imagine a scenario where a set of agents, want to fund a public excludable good, usually known as a club good. This could be anything from a community-funded park with an entrance fee, a digital platform available only to subscribers, a network-attached storage, or deploying and maintaining a smart contract in a Blockchain where state rental is implemented. The central question is: how do these individuals, or agents with private valuations, collaboratively finance such a good in a setting that is permissionless, maintaining the anonymity of the participants? Moreover, what if we want to recover the costs of the public good? This paper explores suitable mechanisms for this task. In permissioned settings, existing mechanisms are studied in the cost-sharing mechanisms literature, like the Shapley value mechanism that offers a solution that recovers the cost of the public good, and is approximately efficient in terms of social cost. But, how do these mechanisms fare in environments where agents operate under pseudonyms? Here, the challenge intensifies as agents might exploit this anonymity to create multiple false identities (known as Sybil or false-name strategies), influencing the outcome to their advantage. This raises two critical questions that we will answer in this paper: 
\begin{enumerate}
    \item Are there mechanisms robust against such strategies while still maintaining efficiency?
    \item What are the implications in terms of efficiency when agents use Sybil strategies in equilibrium?
\end{enumerate} 
Formally, a mechanism for public excludable goods, usually named cost-sharing mechanism can be conceptualized as involving a set \( [n]=\{1,...,n\}\) of players and a cost function \( C: 2^{[n]} \rightarrow \mathbb{R}_+ \). This function models the cost incurred by deploying the public good as a function of the set of agents that have access to it.  Each player \( i \) in this set has a private, non-negative value \( v_i \) for winning, reflecting their valuation for having access to the good or service in question.
In the realm of public excludable goods, the problem is two-fold: determining whether to finance a public good and, if so, identifying the users who are granted access and how much the users have to pay. In the first part of this work, we will focus on the public excludable good problem that is represented by a cost function \( C \) where \( C(\emptyset) = 0 \) and \( C(S) = 1 \) for every non-empty subset \( S \) of players, and in the second part of the paper, we will focus on monotone symmetric submodular cost functions.

The problem of finding, individually rational, truthful, and optimal worst-case social cost is solved in the literature for public excludable goods. However, the advent of the digital age introduces additional complexities, notably the issue of identity misrepresentation. Players capable of creating fake identities or bids can potentially exploit these mechanisms. This phenomenon has been explored in two primary strands of literature: ``Sybil attacks'' and false-name proof mechanisms.

Our research addresses the intersection of these identity dynamics with cost-sharing mechanisms for public excludable goods. We reveal the susceptibility of conventional mechanisms, such as the Shapley value mechanism, to Sybil attacks and establish constraints on social costs for mechanisms satisfying properties like Sybil-proofness.

Building upon these insights, we introduce a comprehensive framework for analyzing cost-sharing mechanisms in the context of public excludable goods, particularly focusing on Sybil strategies and scenarios with an indeterminate number of agents. We introduce the novel concept of Sybil Welfare Invariant mechanisms, distinguishing them from Sybil-proof mechanisms. These mechanisms maintain economic efficiency outcomes, even when agents deploy Sybil or false-name strategies, ensuring robustness against deceptive behaviors by preserving welfare outcomes, irrespective of the number of participants, real or fictitious.

As we delve deeper into our analysis, a notable finding is that the Shapley value cost-sharing mechanism, despite its vulnerabilities, adheres to this property, ensuring that its welfare outcomes are \(\mathcal H_n\)-approximated, thus establishing an upper bound on the worst-case social cost.

We  make several key contributions to the field:
\begin{itemize}
    \item We introduce a framework to analyse cost-sharing mechanisms for public excludable goods with an unknown number of agents and Sybil strategies.
    \item We prove that the cost-sharing mechanism for public excludable goods with constant cost functions such as the Shapley value mechanism, the VCG mechanism for public excludable goods, and the potential mechanism are not truthful under Sybil strategies (i.e. are not Sybil-proof). 
    \item Moreover, in theorem \ref{theorem:lower} we prove that deterministic cost-sharing mechanisms with constant cost functions that are anonymous, individually rational, upper semicontinuous, and Sybil-proof have worst-case welfare social cost $\Omega(n)$. Moreover, if the mechanism has no-deficit, it holds that worst-case welfare social cost is greater than $(n+1)/2$.
    \item Finally, we introduce Sybil welfare invariant mechanisms, and we prove that the Shapley value mechanism for non-decreasing symmetric and submodular cost-functions holds this property. 
\end{itemize}
Conclusively, our research highlights a significant application of the Shapley value mechanism in the context of decentralized systems like Peer-to-Peer (P2P) Networks and Decentralized Finance (DeFi) platforms. We establish that, despite uncertainties in the number of participating agents, the Shapley value mechanism demonstrates consistent worst-case welfare outcomes. This finding is particularly relevant for decentralized autonomous organizations (DAOs) considering the deployment of public excludable goods.
\subsection{Organization of the paper}
The paper will be organized as follows. In Section \ref{preliminars}, we will introduce tools from mechanism design, cost-sharing mechanism, weak dominant strategy sets and, Bayesian games with private beliefs. In Section \ref{spcost}, we will define an anonymous mechanism and the Sybil extension of single-parameter mechanisms. Moreover, we will prove that some of the most important cost-sharing mechanisms for public goods are not Sybil-proof. Moreover, we generalize it by computing the worst-case welfare of Sybil-Proof mechanism for public excludable goods under some mild conditions. In Section \ref{section:SWI}, we introduce the concept of Sybil welfare invariant mechanisms, and we prove that the Shapley value mechanism holds this property. Section \ref{section:conclusions}, we present the conclusions of our study and propose directions for future work. Finally, the appendix contains some proofs the results stated in the paper and the notation used.

\subsection{Related work}
This paper explores the nuanced domain of cost-sharing mechanisms of public excludable goods, particularly emphasizing their non-Sybil proofness guarantees, also known in the literature as false-name proofness. The exploration into this area is rooted in the fundamental work on cost-sharing for public goods and services, where the objective is to allocate costs efficiently among participants. A landmark contribution in this field was made by Moulin and Shenker \cite{moulin2001strategyproof}, who discussed the strategy-proof sharing of submodular costs, with budget-balance constraints. 
This work, alongside Myerson's seminal study \cite{myerson1981optimal} on optimal auction design, forms the bedrock of our understanding of mechanisms that incentivize truthful behavior in cost-sharing scenarios. 
The seminal work of Moulin and Shenker through the Shapley value mechanism \cite{moulin1999incremental,moulin2001strategyproof} initiated a rich vein of research into the efficiency loss of budget-balanced cost-sharing mechanisms. Subsequent studies by Feigenbaum et al. \cite{feigenbaum2003hardness} and Roughgarden et al. \cite{roughgarden2009quantifying}, along with others \cite{bleischwitz2008group,brenner2007cost,chawla2006optimal,gupta2015efficient}, expanded the understanding of these mechanisms under various constraints and objectives. More aligned with our paper, in \cite{dobzinski2008shapley}, the authors proved that no deterministic and budget-balanced cost-sharing mechanism for public excludable good problems that satisfies equal treatment is better than $\mathcal H_n$-approximate.

Regarding Sybil attacks \cite{douceur2002sybil}, where a single malicious entity creates multiple fake identities, pose a significant threat across different domains, from peer-to-peer networks \cite{dinger2006defending,so2011defending} and online social networks \cite{yu2006sybilguard,yu2008sybillimit} to blockchain systems \cite{zhang2019double}, each facing unique challenges due to these attacks.

In the realm of game theory and auction theory, these attacks translate into false-name strategies or shill bids, a concept thoroughly explored in literature. Pioneering work by Yokoo et al. \cite{yokoo2001robust,yokoo2004effect} exposed the vulnerability of Vickrey–Clarke–Groves (VCG) mechanisms to such strategies, marking a significant milestone in understanding the intricacies of auction systems under false-name bids. This line of research was furthered by studies on the efficiency of Sybil-proof combinatorial auction mechanisms \cite{iwasaki2010worst} and the strategic dynamics of shill bidding \cite{sher2012optimal}.

Parallel to these advancements, research in non-monetary mechanisms and voting systems, such as the facility location problem \cite{todo2011false} and voting rules with costs \cite{wagman2008optimal,fioravanti2022false}, has been instrumental in characterizing and tackling Sybil-proof mechanisms. These studies highlight the pervasive nature of Sybil strategies across various decision-making and resource allocation systems.

In their exploration of Sybil-proof mechanisms, the authors in \cite{mazorra2023cost} have developed a comprehensive framework that is notably adaptable for analyzing Sybil extensions in cost-sharing mechanisms. 
\section{Preliminaries}\label{preliminars}
In the upcoming section, we delve into the fundamental concepts of mechanism design, focusing on private valuations, the public excludable goods problem, and Bayesian games with private beliefs. These foundational topics lay the groundwork for our more advanced discussions in subsequent sections. Specifically, the insights gained from private valuations and the public excludable goods problem will be crucial for understanding the complexities presented in Section \ref{spcost}. Additionally, the concept of Bayesian games with private beliefs will be pivotal in Section \ref{section:SWI}, where we examine the Shapley value mechanism under Sybil strategies. It's important to note that this mechanism is not Sybil-proof, and so we use an alternative approach to equilibrium. This preliminary section is designed to equip readers with the necessary theoretical tools to fully grasp the intricacies and challenges of mechanism design in the context of Sybil strategies and non-Sybil-proof environments.

\subsection{Mechanism Design Basics}\label{section:IntroMechanism}

Mechanism design \cite{myerson1989mechanism}, often referred to as the reverse game theory, is a subfield of economics and game theory that focuses on the design of rules and procedures for making collective decisions. While traditional game theory studies how agents make decisions within given rules, mechanism design is concerned with creating the rules themselves to achieve desired outcomes. The central challenge in mechanism design is to ensure that when each participant acts in their own best interest, the collective outcome is still desirable. This is typically achieved by designing mechanisms that align individual incentives with the desired collective outcome.

A particularly important class of problems in mechanism design pertains to situations where agents have private information and/or valuations, and the mechanism designer wants to elicit this information in a truthful manner. This leads to the study of truthful mechanisms. In such mechanisms, agents find in their best interest to report their private information truthfully, rather than misreporting to manipulate the outcome.

One of the most fundamental settings in this context is the single-parameter domain. In these settings, each agent has a single private value (or parameter) that captures its valuation or cost for some service or good. The mechanism designer's task is to determine which agents receive the service (or goods) and at what prices, based on the reported valuations.

A mechanism in this setting can be formally represented as \(\mathcal{M} = (\mathbf{x}, \mathbf{p})\), where:
\begin{itemize}
    \item \(\mathbf{x}\) is the allocation rule that determines which agents receive the service based on their reported valuations.
    \item \(\mathbf{p}\) is the payment rule that specifies how much each agent pays or receives based on their reported valuations.
\end{itemize}

In mechanism design, especially in the context of auctions and allocation problems \cite{maskin1985auction,krishna2009auction}, it is common to assume that agents have private valuations and quasilinear utilities \cite{birmpas2022cost,dobzinski2008shapley} (this is not always the case, but in this paper, we will restrict to this model). This means that the agent's utility depends linearly on the money plus the agent's valuation times the probability of being allocated.
We will use the following notation:
\begin{itemize}
    \item \( v_i\in\mathbb R\) as agent \( i \)'s valuation for the item or service.
    \item \( x_i(b_i, b_{-i}) \) as the allocation rule which determines the probability that agent \( i \) receives the item or service when they report \( b_i \) and the other agents report \( b_{-i} \). In this paper, we will focus on deterministic mechanisms, and so we will assume that $x_i(b)\in\{0,1\}$.
    \item \( p_i(b_i, b_{-i}) \) as the payment rule which determines how much agent \( i \) has to pay (or receives) when they report \( b_i \) and the other agents report \( b_{-i} \).
\end{itemize}
Then, the quasi-linear utility of agent \( i \) when they report \( b_i \) is given by:
\[ u_i(b_i, b_{-i}) = v_i\cdot x_i(b_i, b_{-i}) - p_i(b_i, b_{-i}) \]

For mechanisms to be effective in single-parameter domains, they must satisfy certain properties. One of the most crucial properties is truthfulness, which ensures that agents have no incentive to misreport their valuations. More formally, a mechanism is said to be \textit{truthful} if and only if every agent maximizes their utility by reporting their true type (or valuation) regardless of what the other agents report. That is, for all agents \( i \) and for any valuation $v_i$ and reports $b_i$, and for all possible reports \( b_{-i} \) of the other agents:
\begin{equation*}
 u_i(v_i, b_{-i}) \geq u_i(b_i, b_{-i})   
\end{equation*}

 The characterization of truthful mechanisms in single-parameter domains is elegantly captured by the following theorem.
\begin{theorem}[Myerson's Lemma, see \cite{myerson1981optimal}]
In single parameter domains a normalized mechanism \(\mathcal{M} = (\mathbf{x}, \mathbf{p})\) is truthful if and only if:
\begin{itemize}
    \item \(\mathbf{x}\) is \textit{monotone}: For all \(i=1,...,n\), if \(b_i' \geq b_i\) and \(\mathbf{x}(b_i,b_{-i})=1\) implies \(\mathbf{x}(b'_i,b_{-i})=1\).
    \item \textit{Winners pay threshold payments}: payment of each winning bidder is \(p_i=\text{inf}\{b_i| \mathbf{x}(b_i,b_{-i})=1\text{ and }b_i\geq0\}\).
\end{itemize}
\end{theorem}
Myerson's Lemma provides a foundational result for the design of truthful mechanisms in single-parameter domains. It offers a clear characterization of the allocation and payment rules that ensure truthfulness, paving the way for the design of efficient and optimal mechanisms in various applications.

\subsection{Weak dominant strategy sets \& Private beliefs}
Not all mechanisms studied in this paper will be truthful. When the mechanism is not truthful, the agents' information about the number of other players, their valuation, and their potential strategies have a strong implication in agents strategies and the outcome of the mechanism. This is the case in many real-world situations, where players might not have complete information about the game or about other players' types, preferences, or payoffs. Bayesian games, introduced by John C. Harsanyi \cite{harsanyi1968part}, are a class of games that model such situations of incomplete information. In some cases however, seems unrealistic to exist common knowledge on people preferences, so
players might not only have private information about their own types but also hold private beliefs about the distributions of other players' types and strategies. This contrasts with the standard Bayesian games where players share a common prior over types. Games in which players do not share common priors and instead have their own subjective beliefs are referred to as games with \textit{heterogeneous beliefs} or \textit{subjective priors}.

In this setting, there is a set of players $\mathcal I$ that is not common knowledge with types $t_i\in T_i$ that can take actions in a topological space $A_i$ for $i\in\mathcal I$. Players have  utility function $u_i:T_i\times\prod_{i\in\mathcal I} A_i\rightarrow \mathbb R$ unknown to other players such that $u_i(t_i,\cdot)$ is upper-semicontinuous for every $t_i\in T_i$. Every player has a private belief distribution $\mathcal D_i$ that models the $i$-th players' belief of other players taking a vector of actions $(a_i)_{i\in\mathcal I}$. More formally, given a set $B\subseteq A_{-i}:=\prod_{j\in\mathcal I\setminus {i}} A_j$, the $i$-th player believes that the probability that the vector of action $a_{-i}$ is in the set $B$ is $\Pr_{\mathcal D_i}[B]$. In this setting, we say that a player is \textit{rational with respect to its private beliefs} $\mathcal D_i$ and type $t_i$ if they choose strategies in 
\begin{equation*}
    \underset{a_i\in A_i}{\text{argmax}}\quad\mathbb E_{a_{-i}\sim\mathcal D_i}[u_i(t_i,a_i,a_{-i})].
\end{equation*}
In other words, every player best responds to his/her type and information about the behaviour of the remaining players, while this information can be partial, distorted, or ambiguous \cite{wiszniewska2016belief,dufwenberg2002existence}.
Under some conditions, if the player has type $t_i$, we can restrict that maximization problem to a subset $\emptyset \not=B_i\subsetneq A_{i}$. When $B_i(t_i)$ holds the following property, for every  action $a_i\in A_i$, there is an action $a_i'\in B_i(t_i)$ such that $u_i(t_i,a_i',a_{-i})\geq u_i(t_i,a_i,a_{-i})$ for every action profile $a_{-i}\in A_{-i}$. In this scenario, we say that $B_i(t_i)$ is a \textit{subweak dominant strategy set}. Moreover, if 
\begin{enumerate}
    \item for every $a_i\in A_i\setminus B_i(t_i)$, there is an $a_i'\in B_i(t_i)$ such that $u_i(t_i,a_i',a_{-i})\geq u_i(t_i,a_i,a_{-i})$ for every action profile $a_{-i}\in A_{-i}$ and the inequality is strict for some $a_{-i}\in A_{-i}$, i.e. there exists $a_i'\in B_i(t_i)$ that weakly dominates $a_i$,
    \item for every $a_i\in B_i(t_i)$, there is no $a'_i$ such that weakly dominates $a_i$.
\end{enumerate}
 we say that $B_i(t_i)$ is a \textit{weak dominant strategy set}. An example of (strictly) dominant strategy set on a mechanism is the set $B_i(v_i)=\{v_i\}$ in a  second price auction with private valuations where $v_i$ is the valuation of the item of player $i$.
\begin{lemma}\label{lemma:strict} Given a game $([n],A_i,u_i)$ that for every type $t_i$ has a subweak dominant strategy set $B_i$ such that there exists a chain sets $C_1\subseteq C_2\subseteq...$ such that:
\begin{enumerate}
    \item $C_j\cap B_i(t_i)$ is sequentially compact for all $j\in\mathbb N$ and,
    \item $\cup_{i\in \mathbb N} C_i = A_i$,
\end{enumerate}
then there exists a unique weak dominant strategy set noted by $\overline{B_i(t_i)}$.
\end{lemma}
The proof of the lemma utilizes various technical details from elementary topology, and the reader can verify the proof of the lemma in the appendix.
\begin{definition} A \text{mixed Nash equilibrium with private beliefs} for short, NEPB is a tuple of distributions $(d_1,...,d_n)$ that depends on the tuple of types $(t_1,...,t_n)$ over the set of strategies such that there exists a tuple of distributions $(\mathcal D_1,...,\mathcal D_n)$ that hold:
\begin{equation*}
        \mathbb E_{(a_i,a_{-i})\sim d_i\times\mathcal D_i}[u_i(t_i,a_i,a_{-i})] =\underset{a_i\in A_i}{\text{max}}\quad\mathbb E_{a_{-i}\sim\mathcal D_i}[u_i(t_i,a_i,a_{-i})] .
\end{equation*}
If the distributions $\mathcal D_1,...,\mathcal D_n$ have full support, we say the tuple of distributions $(d_1,...,d_n)$ is mixed \textit{Nash equilibrium with full support private beliefs} (NESPB).When all elements $(d_1,...,d_n)$ are atomic with one element (i.e. $\Pr_{d_i}[a_i]=1$ for some $a_i\in A_i$), we say that the equilibrium is \textit{pure}.
\end{definition}
Note that the tuple $(d_1,...,d_n)$ is a NESPB, and player $i$ has a weak dominant strategy set $B_i(t_i)$ then $\Pr_{d_i}[B_i(t_i)]=1$. Also, the reader should note that this notion of equilibrium is very weak. Every mixed Nash equilibrium and Bayes Nash equilibrium with common priors are mixed Nash equilibrium with private beliefs. When restricting to the set of NESPB, fundamentally, the only thing that we are imposing for a set of strategies to be a subjective equilibrium is that rational agents take actions from the weak dominant strategy set. This is implicitly done in the literature. This condition is also known as the \textit{no-overbidding} assumption \cite{roughgarden2017price}. For example, when a mechanism is truthful, even if reporting truthful valuations to the mechanism is not strictly dominant, but just weakly dominant, it is assumed that agents report their valuations truthfully to the mechanism. Analogously, we extend this idea to non-truthful mechanisms with weak dominant strategy sets.
These concepts will be central for the definition of Sybil welfare invariant mechanisms in section \ref{section:SWI}.

\subsection{Public excludable goods}
In this section we introduce the notation and key concepts proposed in \cite{dobzinski2008shapley} within the context of cost-sharing mechanism design.

In a  cost-sharing mechanism design problem \cite{dobzinski2008shapley,sundararajan2009trade}, several participants with unknown preferences vie to receive some goods or services, and each possible outcome has a known cost. Formally, we have a service and every player $i\in [n]$ has a valuation function $v_i\in \mathbb R_+$. The assumption here is that there are no externalities; each player's value is purely determined by the goods they receive, irrespective of other players accessing the same services. There is a cost function $C:2^{[n]}\rightarrow\mathbb R_+$ that specifies the costs of every possible allocation of services.

The section \ref{spcost} focuses on the study of the \textit{public excludable good} problem, which involves determining whether to finance a public good and, if so, identifying who is allowed to use it. 
The public excludable good problem has cost function $C(S)=c\in\mathbb R_+$ (wlog in this paper we will assume that $c=1$), for every $S\not=\emptyset$ and $C(\emptyset)=0$. In section \ref{section:SWI} we will study cost-sharing mechanisms with monotone, symmetric and submodular functions $C$. That is, we assume that there exists a non-decreasing concave function $f:\mathbb R_+\rightarrow\mathbb R_+$ such that $f(0)=0$ and $C([n]) =f(n)$ for all $n\in\mathbb N$. 

 A (deterministic) \textit{cost-sharing mechanism} consists of an allocation rule $\textbf{x}:\mathbb R^n_+\rightarrow \{0,1\}^n $ and a payment rule $\textbf{p}:\mathbb R^n_+\rightarrow\mathbb R^n$ that for a reported bid vector profile $\bb=(b_1,...,b_n)$ will determine the set of allocated agents $S =\{i\in [n]:x_i(\bb)=1\}$, and $p_i\geq0$ is player $i'$s payment.  We assume that players have quasi-linear utilities, meaning that each player $i$ aims to maximize $u_i(\bb,p)=x_i(\bb)v_i-p_i(\bb)$.

In this paper, we will always require the following standard axiomatic properties:
\begin{itemize}
    \item \textit{No positive transfers} (NPT): Players never get paid, i.e., $p_i(\textbf{b})\geq0$.
    \item \textit{Individual rationality} (IR): If the allocation is $(S,p)$, players never pay more than they bid, otherwise, they are charged nothing, i.e., $p_i(\textbf{b})\leq x_i(\bb)b_i$.
    \item \textit{Anonymity/Symmetry}: For any permutation $\sigma\in S_{n}$, it holds $x_i(\textbf{b})=x_{\sigma(i)}(\sigma(\textbf{b}))$ and $p_i(\textbf{b})=p_{\sigma(i)}(\sigma(\textbf{b}))$.
    \item \textit{No-Deficit}: For an allocation $(S,p)$, the sum of payments exceeds the costs incurred by providing the service, i.e. $\sum_{i=1}p_i(\textbf{b})\geq C(S)$.
    \item $\beta$-\textit{Budget-balance} for $\beta\geq 1$ if for every allocation $(S,p)$ if $C(S)\leq\sum_{i=1}^n p_i(\textbf{b})\leq \beta C(S)$. In case, that a mechanism is $1-$budget balance, it is said that the mechanism is budget-balanced.
    \item \textit{Truthful}: Following the definition of truthfulness made in \ref{section:IntroMechanism}, a cost-sharing mechanism is truthful if for every bid valuation vector $\textbf{b}_{-i}$, true valuation $v_i$ and reported valuation $b_i\in V_i$, holds 
    \begin{equation}\label{eq:sp}
        x_i(v_i,\textbf{b}_{-i})v_i-p_i(v_i,\textbf{b}_{-i})\geq x_i(b_i,\textbf{b}_{-i})v_i-p_i(b_i,\textbf{b}_{-i}).
    \end{equation}
    where $S$ is the allocation with reports $v_i,\bb_{-i}$ and $S'$ is the allocation with reports $\bb$.
    \item \textit{Consumer sovereignty} (CS): For all players $i$ and bids $\bb_{-i}$, there exists a bid $b_i$ such that player $i$ has access to the public good when the bid profile is $(b_i,\bb_{-i})$.
    \item \textit{upper semi-continuity}: For every player $i$ and bids $b_{-i}$ of the other players: if player $i$ wins with every bid larger than $b_i$, then it also wins with the bid $b_i$.
\end{itemize}
Another weaker version of anonymity \cite{dobzinski2008shapley} is the following:
\begin{itemize}
    \item \textit{Equal treatment} (EQ): Every two players $i$ and $j$ that submit the same bid receive the same allocation and price.
\end{itemize}
Furthermore, we assert the following technical property: if all players are served with a specific bid, then those same players are also served for all bids that are larger in their component. Formally,
\begin{itemize}
\item  \textit{Monotonicity}: For all $i$, if $x_i(b_i,\bb_{-i})=1$, then for all $b'_{i}\geq b_{i}$, $x_i(b'_i,\bb_{-i})=1$.
\end{itemize}
Another stronger version of monotonicity is the following one that states that for every player $i$, will not have a negative impact on the allocation of the public good if some players increase their bid. More formally:
\begin{itemize}
    \item \textit{Strong-monotonicity}: For all $i$, if $x_i(\bb)=1$, then for all $\bb'\geq \bb$ (that is $b'_i\geq b_i$ for all $i=1,...,n$) holds $x_i(\bb')=1$.
\end{itemize} 
This principle implies that an increased valuation of the public good by any number of players will not negatively influence the allocation outcome for all involved players. Essentially, strong monotonicity ensures that higher valuations by some players don't lead to a disadvantageous allocation for others.
Another (stronger) version of strategy-proofness also includes the notion of a mechanism being resistant to coordinated manipulation by users or in other words, preventing users to have incentives to collude in order to individually maximize their utility.
\begin{itemize}
    \item  A cost sharing mechanism is \textit{group strategy-proof} (GSP) if for all true valuations $v\in\mathbb R^n_+$ and all non-empty coalitions $K\subseteq [n]$, there is no $\textbf{b}$ such that $\textbf{b}_{-K}=\textbf{v}$ with $u_K(\textbf{b})>u_K(\textbf{v})$.
\end{itemize}
Moulin et. al. \cite{moulin1999incremental} proved that if a mechanism $\mathcal M$ is an upper-semi continuous and group-strategy proof then the mechanism is \textit{separable}, i.e. all players that have access to the public just depends on the set (and not the bid). More formally \cite{sundararajan2009trade}:
\begin{itemize}
    \item A \textit{cost-sharing method} is a function $\zeta:2^{[n]}\rightarrow\mathbb R^n_{\geq0}$ that associates each set of players to a cost distribution, where for all $S\subseteq [n]$ and all $i\not\in S$ it holds that $\zeta_i(S)=0$. A cost-sharing mechanism $\mathcal M = (\x,\p)$ is \textit{separable} if there exists a cost-sharing method $\zeta$ such that $\p(\textbf{b})=\zeta(S(\textbf{b}))$, where $S(\textbf{b}) = \{i\in[n]:x_i(\textbf{b})=1\}$.
\end{itemize}
Observe that if the mechanism is separable and symmetric and $C$ is symmetric (i.e. $C(S)$ just depends on the number of elements in $S$) then $\zeta_i(S)=\zeta_j(S)$ for all $i,j\in S$. 

Now, for economic efficiency, the service cost and the rejected players' valuations should be traded off as good as possible. A measure for this trade-off is the \textit{social cost} of function $\pi:2^{[n]}\rightarrow \mathbb R_{\geq0}$. Given the cost $C$ and the true valuations functions $v_1,...,v_n$, social costs are defined by 
\begin{equation*}
    \pi(S):=C(S)+\sum_{i\not\in S} v_i.
\end{equation*}
That is, the social cost of an allocation $S$ is the cost of granting access to the public good to $S$ players, and the valuations of the agents that do not have access to the public good.

A mechanism is said to be $\alpha$-approximate \cite{dobzinski2008shapley}, with respect to the social cost objective if for every tuple of valuations $(v_1,..,v_n)$, the allocation $S$ of the mechanism satisfies
\begin{equation}
    \pi(S)\leq \alpha \pi(S^\star)
\end{equation}
where $S^\star$ is the optimal allocation, that is, the allocation that minimizes the social cost. As an observation, a mechanism is $\alpha-$approximate for some $\alpha\in \mathbb R^+$ if and only if the mechanism holds the consumer sovereignty property \cite{dobzinski2008shapley}\footnote{One of the implications is not proved explicitly in the paper but the argument of the proof is fundamentally the same to the other implication.}. It is known that there are nor better than $\mathcal{H}_n$-approximated randomized or deterministic, truthful and no-deficit mechanism, where $\mathcal H_n=\sum_{i=1}^n\frac{1}{i} =\Theta(\log(n))$.

Let us justify why it is more economically efficient to fund excludable goods than non-excludable ones with private valuations. If we place the constraint that all or none must be served the public good (making the public good a not excludable good), we have that all no-deficit truthful mechanisms are at least $(n-1+1/n)-$worst-case welfare. Let's prove it. Assume that the cost of the public good is $c=1$ without loss of generality. Now, suppose that there is a mechanism $(\x,\p)$ with this constraint. Assume that all players have valuation $v_i=1-\varepsilon$ and suppose that all agents have access to the public good, for sufficiently small $\varepsilon$ (otherwise, the mechanism has worst-case welfare lower bounded by $n$ finishing the proof). Let $p_1,...,p_n$ be their respective payment. Since the mechanism has no-deficit, we know that $\sum_{i=1}^n p_i\geq 1$. And so, there exists $i$ such that $p_i\geq 1/n$. Therefore, the vector profile $(1-\varepsilon,...,1/n-\varepsilon,...,1-\varepsilon)$ has the empty allocation, leaving to a social cost $n-1+1/n-n\varepsilon$. Making $\varepsilon$ converge to zero, we obtain that the mechanism is at least $(n-1+1/n)-$worst-case welfare.

\section{Sybil-Proof and Cost-sharing mechanisms}\label{spcost}
Note that the given definitions in the preliminaries presume the mechanism designer is aware of the number of identities. Additionally, agents can only submit a bid to the mechanism without the ability to create or alter other identities to influence the mechanism. However, in general, this is not true. For example, one can create multiple identities to access a social network, multiple bank accounts to bid in an ad auction \cite{varian2009online} or use multiple public keys to interact with a DeFi protocol \cite{mazorra2022price,mazorra2023cost}. This provides a new challenge and problems that are worth studying in permissionless mechanism design. For example, in \cite{dobzinski2008shapley} the authors analyze two different truthful mechanisms for public excludable good. The VCG mechanism and the Shapley value mechanism. For completeness, we will write these mechanisms in this section. The first one is efficient (welfare maximizer) however, in general, has deficit (i.e. the users' payments do not cover the costs incurred by financing the public good). The second one, the shapley value mechanism, is budget-balance, however, has $\mathcal H_n$-approximately welfare, where $\mathcal H_n$ are the harmonic numbers. Moreover, the authors prove that this mechanism is worst-case welfare optimal (up to a constant) in the set of truthful, incentive-compatible, budget-balance, and equal treatment mechanisms. However, none of these mechanisms are truthful in the permissionless setting. In other words, the mechanisms are not Sybil-proof/false-name proof. 

In this section, we will discuss the public excludable good in permissionless settings. First, we will formalize the public excludable good problem with an unknown number of identities where agents can use Sybils to maximize their payoff. Then, we will prove that the VCG mechanism, the Shapley value mechanism and the potential mechanism are not Sybil-proof. Moreover, we will generalize it by proving that all no-deficit, upper semicontinuos and truthful mechanisms are at least $(n+1)/2-$approximated. In other words,  we will effectively establish both upper and lower bounds for the worst-case social cost match. This finding underscores a critical trade-off: achieving a completely strategy-proof (truthful, resistant to Sybil attacks, and immune to group strategies) cost-sharing mechanism necessitates a compromise in terms of economic efficiency.

\subsection{Sybil extension of Cost-sharing mechanisms}

In general, truthfulness captures the idea that players cannot act strategically in order to obtain more utility from the mechanisms. However, in pseudo-anonymous environments such as blockchain, this is not necessarily true \cite{mazorra2023cost}, since players can create multiple identities and strategically manipulate the outcome of a truthful mechanism. When agents have no incentives to create multiple identities, we say that the mechanism is Sybil-proof or false-name proof \cite{mazorra2023cost}. To define it, we must extend the definition of a mechanism when 1) the number of identities is unknown 2) users can use more than one identity. This is presented in \cite{mazorra2022price} and is called the Sybil extension mechanism. 
 First, we have to define the mechanism with unbounded but finite number of players, called \textit{anonymous mechanism}. An anonymous mechanism $\mathcal M$ is a sequence of maps $\{(\textbf{x}^n:\Rr_+^n\rightarrow \{0,1\}^n,\textbf{p}^n:\Rr_+^n\rightarrow \Rr_+^n )\}_{n\in\mathbb N}$ such that the following two properties hold:
\begin{itemize}\label{axioms}
    \item \textit{Anonymity}: The maps $\textbf{x}^n$ and $\textbf{p}^n$ are equivariant under the action of $S_n$, that is, for all $\sigma\in S_n,\,b\in \Rr_+^n$, $\textbf{x}(\sigma b) = \sigma \textbf{x}(b)$ and $\textbf{p}(\sigma b) = \sigma \textbf{p}(b)$.
    \item \textit{Consistency}: Let $i_{n,m}:\Rr_+^n\hookrightarrow \Rr_+^m$ be any inclusion map that comes from taking the identity map on the first $n$ components and zero-filling the remaining $m-n$ components and permutating the $m$ components by a permutation of $S_m$. Let $p_{n,m}:\Rr_+^m\rightarrow \Rr_+^n$ be the projection such that $p_{n,m}\circ i_{n,m}= id_{\Rr_+^n}$. Then, the following diagram commutes\footnote{\textbf{Category theory observation}: If we take $A_n$ to be the set of symmetric mechanisms with $n$ agents, and $f_n:A_n\rightarrow A_{n-1}$ to be $(\x^n,\p^n)\mapsto(\x^n\circ i_{n-1,n},\p^n\circ i_{n-1,n})$, then anonymous Sybil mechanisms are elements of the inverse limit
\begin{equation*}
    \varprojlim A_i = \left\{ (a_i)_{i \in \mathbb N} \mid a_i \in A_i \text{ for all } i \in \mathbb N \text{ and } f_{ij}(a_j) = a_i \text{ for all } i \leq j \right\}.
\end{equation*}}:

    \begin{minipage}{0.5\textwidth}
    \hspace{1cm}
        \begin{tikzcd}
            \Rr_+^n \arrow[d, "i_{n,m}", hook, swap] \arrow[r, "\x^n", swap] \arrow[rd, phantom, "\circlearrowleft" description] & \{0,1\}^n \arrow[d, "i_{n,m}", hook] \\
            \Rr_+^m \arrow[r, "\x^m", swap] & \{0,1\}^m
        \end{tikzcd}
    \end{minipage}
    \begin{minipage}{0.5\textwidth}
    \begin{tikzcd}
        \Rr_+^n \arrow[d, "i_{n,m}", hook, swap] \arrow[r, "\p^n", swap] 
        \arrow[rd, phantom, "\circlearrowleft" description] & \Rr_+^n \arrow[d, "i_{n,m}", hook] \\
        \Rr_+^m \arrow[r, "\p^m", swap] & \Rr_+^m
    \end{tikzcd}
\end{minipage}
\end{itemize}

Now, given an anonymous mechanism, we can define the \textit{Sybil extension mechanism}. The Sybil extension mechanism of an anonymous mechanism $\{(\x^n,\p^n)\}_{n\in\mathbb N}$ consists of extending action space of each player $i$ from $\mathbb R_+$ to $\mathbb R_+^\infty$. Every agent $i$ can report a finite with arbitrary length set of bids $b_i = (b^{1}_i,...,b^{k_i}_i,0,0,...)$, for some $k_i\in\mathbb N$, and we define $K_i=[k_i]$. Then, we take $m = \sum_{i=1}^{n} k_i$, $\textbf{b}=(b_1,...,b_n)$ and compute $x_j(\textbf{b})$ and $p_j(\textbf{b})$ for every $j=1,...,m$. This is well-defined since the mechanism $(\x^m,\p^m)$ is symmetric. Now, the total payment of player $i$ consists of the sum of payments of its Sybil identities. That is, $\p_i(\textbf{b}) = \sum_{j\in K_i} p^m_j(\textbf{b})$. In this paper, we consider that the players' allocation, is the best allocation among all its Sybils' allocation. That is, $\x_i(b_i) = \max_{k\in K_i}\{x^m_k(b_i, b_{-i})\}$. Therefore, its utility is 
\begin{equation}
 v_i\max_{k\in K_i}\{x^m_k(b_i,b_{-i})\}-\sum_{k\in K_i}p^m_k(b_i,b_{-i}), 
\end{equation}
where $b_{-i}$. We note the mechanism $(\x,\p)$ as the Sybil extension mechanism and denote it by $\textbf{Sy}(\mathcal M)$. Observe that in this definition of agents utility in the Sybil setting, we are assuming that the agents have no extra utility to have more Sybil identities to have access to the public good. The motivation for this particular type of Sybil extension stems from the nature of the mechanisms involved, which are geared towards funding and using public goods. In these scenarios, agents gain no additional utility from possessing multiple identities. This is because access to the public good is not enhanced by having more than one identifier. In simpler terms, if the mechanism employs a whitelisting process for identifiers, having more than one does not provide any extra benefit to the users. Once a user has one identifier, they already have the necessary access to utilize the public good. Thus, multiple identifiers do not translate into increased utility in the context of these public goods funding mechanisms. In this context, we say that an anonymous mechanism is Sybil-proof if no agents have incentives to create Sybil identities to increase its payoff. More formally:
\begin{itemize}
    \item  An anonymous cost-sharing mechanism is \textit{Sybil-proof} if for every player $i$, every vector $b_{-i}\in \mathbb R_+^\infty$ of bids,  bid vector $b_i\in\mathbb R_+^{\infty}$ with sybils $K_i$, we have that
    \begin{equation}
        \hspace{-0.5cm}
         v_i\textbf{x}_i(v_i,\textbf{b}_{-i})-\textbf{p}_i(v_i,b_{-i})\geq v_i\max_{k\in K_i}\{x^m_k(b_i, b_{-i})\}-\sum_{k\in K}p_k (b_i,b_{-i}).
    \end{equation}
\end{itemize}
In other words, the mechanism is Sybil-proof if the Sybil extension mechanism is truthful.

\subsection{Limits of Sybil-Proof cost-sharing mechanisms}
We will see next that the mechanisms proposed in \cite{dobzinski2008shapley} and \cite{dobzinski2017combinatorial} are not Sybil-proof. The mechanisms proposed in \cite{dobzinski2008shapley} are the VCG mechanism applied to the excludable public good problem and the Shapley value mechanism. On the other hand, the mechanism proposed in \cite{dobzinski2017combinatorial} is a modified version of the VCG mechanism with adding a cost per user. In the following, we will describe these mechanisms for completeness.
 
\begin{mybox2}{Cost-sharing mechanism for public excludable goods}
\textbf{Input}: Bids $b_1,...,b_n$.\\
\textbf{Output}: The set of agents $S^\star$ that are served and the payment vector $p=(p_1,...,p_n)$.\\
\textbf{VCG mechanism}
    \begin{enumerate}
    \item Choose the outcome $S^\star=[n]$ if $\sum_{i=1}^n b_i>1$ and $S^\star=\emptyset$ otherwise.
    \item Charge each player $i\in [n]$ the amount $p_i=\max\{0,1-\sum_{j\in [n]\setminus i} b_j\}$.
    \end{enumerate}
\textbf{Shapley value mechanism for submodular monotone cost functions}
    \begin{enumerate}
    \item Order the bids in descending order, wlog $b_1\geq b_2\geq... \geq b_n$.
    \item Take $k=\text{argmax}_i\{b_i\geq C([i])/i\}$.
    \item The players $i=1,...,k$ have access to the public good, i.e. $S^\star = [k]$ and each player pays $p_i =C([k])/k$ for $i=1,...,k$.
    \end{enumerate}
\textbf{Potential mechanism}
    \begin{enumerate}
    \item Choose the  outcome $S^\star\in \text{argmax}_{S\subseteq [n]}\left\{ \sum_{i\in S} b_i - \mathcal H_{\mid S\mid}\right\}$.
    \item Charge the players $i\in S^\star$ the amount $p_i = [\sum_{i\in S^\star_{-i}} b_i -\mathcal H_{S^\star_{-i}}] - [ \sum_{i\in S^\star\setminus i} b_i -\mathcal H_{S^\star}]$ and zero otherwise. Where $S^\star_{-i}$ is the set that maximizes the previous function with $b_i=0$.
    \end{enumerate}
\end{mybox2}
\begin{observation} The VCG, Shapley and Potential mechanism for constant cost-functions are symmetric, individually rational, and truthful \cite{dobzinski2008shapley}. The potential mechanism and the Shapley value mechanism have no deficit, while the VCG mechanism has deficit \cite{dobzinski2008shapley,dobzinski2017combinatorial}. Moreover, the Shapley value mechanism is group strategy-proof.
\end{observation}
Before stating the results, we must extend the definition of \( C \) to capture the cost of developing the public good with an arbitrary number of sybils. The extension of the cost function is defined, as one would anticipate, by \( C(S) = C([n]) \) for every \( \emptyset\not=S \subseteq \mathbb{N} \) in case of the public goods problem. Otherwise, we will assume that the cost function is a monotone symmetric function $C:2^{\mathbb N}\rightarrow \mathbb R_+$.

\begin{proposition}\label{prop:cost}  The Shapley value mechanism, the VCG for public excludable goods, and the potential mechanism for constant cost functions are not Sybil-proof.
\end{proposition}
\begin{proof} Wlog we assume that $C$ is constant $1$ for non-empty subsets. We prove it separately for each mechanism.

\textbf{VCG-mechanism}: Assume that there are two players with valuations $v_1=v_2=1/3$. In this case, the public good is not funded and so $S^\star = \emptyset$. In this case, the utility of this outcome is zero for both players. If the first player generates two identities $3$ and $4$ and bids $v_3=v_4 =1$, then the allocation is $S^\star=\{1,2,3,4\}$ and the payment of all players is $0$. In this case, the utility of the player is $1/3$ and so the mechanism is not Sybil-proof. 

\textbf{Shapley value mechanism}: Now, let's assume that there are three players with valuations $v_1=1+\varepsilon$ and $v_2=v_3=1/3-\varepsilon$. Then, the outcome of the mechanism is $S^\star=\{1\}$ with $p_1=1$, and so, has utility $\varepsilon$. On the other hand, player $1$ splits its bid in two $b_1=1/4$ and $b_1=1/4$, and the outcome is $S^\star=\{1,2,3,4\}$ with total payment $p =  1/2$, therefore the total utility $1/2+\varepsilon$. Therefore, the mechanism is not Sybil-proof.

\textbf{Potential mechanism}: Assume that the valuations are $v_1=1+\varepsilon$ and $v_i = 1/i-\varepsilon$ for $i=2,...,n$. Then, the allocation of the potential mechanism is $S^\star = \{1\}$ and payment $p_1=1$. If the first player creates a Sybil with valuation $v_{n+1}=1+\varepsilon$, then, in this setting, the allocation is $S^\star = [n+1]$. Now, let us compute the payment. We have that $S^\star_{-1} = \{n+1\}$ and $S^\star_{-(n+1)}= \{1\}$. Therefore, the payments are:
\begin{equation*}
    p_{n+1} = -[1+\varepsilon+\sum_{i=2}^n (1/i - \varepsilon) -\mathcal H_{n+1}] =1/(n+1)-(n-2)\varepsilon 
\end{equation*}
and analogously $p_1=p_{n+1}$. When $\varepsilon$ tends to $0$, with the Sybil strategy, the utility of the player is $1- 1/(n+1) > 0$, making the strategy profitable.
\end{proof}
Observe that this is not true for some classes of cost functions. For example, for symmetric additive cost functions, i.e. $C(S) = |S|$ the Shapley-value mechanism is Sybil-proof and welfare maximizing. The proof is actually simple. The Shapley value mechanism with additive valuation allocated the public good to the bidders with cost larger than $1$ and the payment is $1$ for each player, therefore making Sybils increase the payments. And so, in this section, we will focus on constant cost functions.

The last proposition opens the following question. What is the minimum $\alpha(n)$ such that there is an $\alpha(n)-$approximated truthful, and Sybil-proof cost-sharing mechanism? By \cite{dobzinski2008shapley}, we know that the unique truthful, incentive-compatible, budget-balance, equal treatment, and upper continuous mechanism is the Shapley Value mechanism. Therefore, to have Sybil-proof mechanism, we will have to sacrifice at least one of the previous properties. If we sacrifice budget-balance, we will obtain the Optimal Sybil-Proof mechanism (the use of the word ``optimal" will be clear by the Theorem \ref{theorem:lower}). 

\begin{mybox2}{Optimal Sybil-Proof mechanism}
    \begin{enumerate}
    \item Accept bids $b_1,...,b_n$.
    \item Order the bids in descending order, wlog $b_1\geq b_2\geq... \geq b_n$.
    \item Take $k=\text{argmax}_i\{b_i\geq C([n])/2\}$.
    \item If $k=1$, then do not allocate the public good to any player and set the payments to zero, unless $b_1\geq C([n])$, then serve the public good to the first player and set $p_1=C([n])$. Otherwise, the players $i=1,...,k$ have access to the public good and each player pays $p_i =C([n])/2$ for $i=1,...,k$, and $p_i=0$ for the remaining players.
\end{enumerate}
\end{mybox2}
\begin{proposition}\label{theorem:SPHybrid} The Optimal Sybil-Proof mechanism for constant cost functions is individually rational, group-strategy proof, Sybil-proof, no-deficit, and $(n+1)/2-$approximate.
\end{proposition}
\begin{proof} Again, wlog we assume that $C(S)=1$ for all $S\not=\emptyset$. Observe that the mechanism is separable with function $\zeta(S) = \begin{cases}\frac{1}{2},\text{ if } |S|\geq2,\\1,\text{ if }|S|=1,\\ 0, S=\emptyset\end{cases}$ and so is incentive compatible and group strategy-proof. See \cite{sundararajan2009trade} for more details. The no-deficit condition follows by construction. The worst-case welfare can be proved by considering the bid vector profile $(1-\varepsilon,1/2-\varepsilon,...,1/2-\varepsilon)$ for $\varepsilon>0$. Observe that no player is allocated and so the social cost is $(n+1)/2-n\varepsilon$. Making $\varepsilon\rightarrow 0$, we obtain that the worst-case social cost is lower bounded by $(n+1)/2$. Now, lets prove that the worst-case welfare is upper bounded by $(n+1)/2$. Lets $v_1\geq...\geq v_n$ be a bid vector profile. Suppose that the first $k$ are allocated the public good. If $k=0$, then $v_1<1$, and $v_2,...,v_n<1/2$, therefore the social cost is upper bounded by $1+(n-1)/2$. If $k\geq 1$, then $n-k$ players do not have access to the public good, and so, having valuations $v_i<1/2$ for $i=k+1,...,n$. And so, again the social cost is bounded by $(n+1)/2$. This concludes the proof that the mechanism is $((n+1)/2)-$approximate. The mechanism is clearly Sybil-proof. Given a vector of reports $v_1,...,v_n$, if just one is allocated, then its payment is $1$. Making a Sybil will not decrease its payment since at most will decrease the payment of the original identity, from $1$ to $1/2$ but will add the payment of the Sybil identity $C([n])/2$.
    
\end{proof}
 We will prove that this upper bound on social costs matches the lower bound when the mechanism is group strategy-proof and Sybil-proof. In fact, we will prove a more general result, that states that every cost-sharing mechanism that is individually rational, $\alpha(n)-$approximate, symmetric,  truthful, upper semicontinuos, and Sybil-proof, then $\alpha(n)=\Omega(n)$. Moreover, if the mechanism is no-deficit $\alpha(n)\geq (n+1)/2$. In other words, to add Sybil-proofness, we must increase exponentially the worst-case social cost from $\mathcal H_n$ to $(n+1)/2$ with no-deficit mechanism and from $0$ to $\Omega(n)$ for generic mechanisms.

To make the proof more understandable, we will break the proof into a Lemma and a Proposition. In the proof, we assume without loss of generality that $C$ is constant $1$ for non-empty subsets. From now on, we will fix the mechanism $(\x,\p)$ and assume that holds all previous stated properties.

\begin{lemma}\label{lemma:6}\normalfont Let $M$ be an anonymous, truthful, individually rational, and Sybil-proof cost-sharing mechanism for constant cost functions that satisfies upper semi-continuity, then, for every agent $i$ and bid vector profile $\textbf{b}$ such that $x_i(\textbf{b})=1$ and $x_j(\textbf{b})=0$, then $x_i(\tilde{\textbf{b}}_{j})=1$, where $\tilde{\textbf{b}}_j=(b_1,...,b_{j-1},0,b_{j+1},...,b_n)$.
\end{lemma}
\begin{proof}
Assume that there exists a bid vector profile $\textbf{b} = (b_1, \ldots, b_n)$ such that $x_i(\textbf{b}) = 1$ and $x_j(\textbf{b}) = 0$, and also $x_i(\tilde{\textbf{b}}_{j}) = 0$. Given that the mechanism satisfies upper semi-continuity, there exists $\varepsilon > 0$ such that for $\textbf{b}' = (b_1, \ldots, b_i + \varepsilon, \ldots, b_n)$, the aforementioned properties hold. Now, assume that agents have valuations $\tilde{\textbf{b}'}_j$. If agents report their true valuations, then agent $i$ receives a null allocation, resulting in zero utility. However, if agent $i$ reports an additional identity with bid $b_j$, then agent $i$ is allocated, with a payment $p \geq b_i$, leading to a utility of at least $\varepsilon > 0$. This implies that the mechanism is not Sybil-proof, which is a contradiction.
\end{proof}

Now, consider the following sequence:
\begin{equation*}
     v_n = \text{sup}_{v\geq 0}\{v\mid \textbf{x}(v_1-\varepsilon_1,...,v_{n-1}-\varepsilon_n,v) = \vec{0},\text{ for all }\varepsilon_i\in(0,\min\{v_{i}:i,...,n-1\})\}
\end{equation*}
\begin{proposition} Given a $M$ anonymous, truthful, individually rational, and Sybil-proof cost-sharing mechanism for constant cost functions that satisfies upper semi-continuity holds that for all $n$, $\max\{v_m:m\in[n]\}\leq 2v_n$. In particular, $v_n\geq v_1/2$.
\end{proposition}
\begin{proof} Let $m\in \text{argmax}\{v_m:m\in [n]\}$. If $v_m=v_n$, there is nothing to prove, so let's assume that $v_m>v_n$. Now, by definition of $\{v_n\}$ and the upper semicontinuos property of the mechanism, there exist $\varepsilon_1>0,...,\varepsilon_{n-1}>0$, such that the allocation $(v_1-\varepsilon_1,....,v_{n-1}-\varepsilon_{n-1},v_n)$ is not null. 
Now, suppose that there are $n-1$ agents with valuations $\textbf{v} = (v_1-\varepsilon_1,...,v_m-\delta,...,v_{n-1}-\varepsilon_{n-1})$ with $\delta\in (0,\min\{\varepsilon_m,v_m-v_n\})$. 
By lemma \ref{lemma:6} and the definition of the sequence $v_1,...,v_n$ holds $x_n(\textbf{v},v_n)=1$, and so $x_n((\textbf{v},v_m-\delta))=1$ by truthfulness. Taking $\textbf{b}= (\textbf{v},v_m-\delta)$ by symmetry, $x_m(\textbf{b})=1$ and $x_n(v_n,\textbf{b}_{-m})=1$. And so the payment of the identities $n$ and $m$ with the bid vector profile $\textbf{b}$ is at most $v_n$. 
In the case that agents have valuations $\textbf{v}$  and the agent $m$ reports two Sybil identities with valuations $v_m-\delta$, the utility of the agent is $v_m-\delta-2p\geq v_m-\delta -2v_n$. Since the mechanism is Sybil-proof and the allocation with the valuation profile $\textbf{v}$ is empty, making $\delta\rightarrow0$ it holds $0\geq v_m-\delta -2v_n$ and so $2v_n\geq v_m$. 
\end{proof}
Now, we are ready to prove the following theorem by using the previous lemmas.
\begin{theorem}[Social-cost lower bound]\label{theorem:lower} No individually rational, anonymous,  truthful, upper semicontinuos, and Sybil-proof, deterministic cost-sharing mechanism is better than $\Omega(n)-$approximate. Moreover, if the mechanism has no-deficit, the mechanism is no better than $(n+1)/2-$approximate.
\end{theorem}
Since for all $\varepsilon>0$ it holds that $x(v_1-\varepsilon,...,v_n-\varepsilon)=\vec{0}$, we have that the social cost is $\sum_{i=1}^n (v_i-\varepsilon)$. Since the mechanism is $\alpha(n)$-approximate, we have that $\sum_{i=1}^n (v_i-\varepsilon)\leq \alpha(n)$, making $\varepsilon\rightarrow0$, we have that $\sum_{i=1}^n v_i\leq \alpha(n)$.
 We know that, $2v_n\geq v_1$ for all $n\in\mathbb N$ so: 
\begin{align*}
    \alpha(n)&\geq \sum_{i=1}^n v_i \\
             &\geq v_1+\sum_{i=2}^n \frac{v_1}{2}\\
             &= \frac{n+1}{2}v_1
\end{align*}
If $v_1=0$, then the mechanism is not $\alpha-$approximate for any finite $\alpha$, since the social cost when one agent reports $b_1=0$ is $1$ and the optimal allocation has $0$ social cost, implying the result. If $v_1>0$, then $
\alpha(n)\geq v_1(n+1)/2$. In either case, $\alpha(n)=\Omega(n)$. If the mechanism is no-deficit, $v_1\geq 1$, and so $\alpha(n)\geq (n+1)/2$. This concludes the proof of the Theorem $\ref{theorem:lower}$. 
Now, since every upper semicontinuous group strategy-proof is separable, we deduce the following corollary.
\begin{corollary}  No incentive compatible, upper semi-continuos, no-deficit, symmetric, group strategy-proof, deterministic, and Sybil-proof cost-sharing mechanism for constant cost functions is better than $(n+1)/2-$approximate.
\end{corollary}

Therefore, by Theorem \ref{theorem:lower} and Proposition \ref{theorem:SPHybrid} we have upper and lower bounded the worst-case welfare match. This result shows that if one wants a truthful, upper semicontinuous, Sybil-proof, cost-sharing mechanism, then the mechanism must sacrifice economic efficiency.

\section{Sybil Welfare invariant mechanisms}\label{section:SWI}

In Section \ref{spcost}, we demonstrated that the Shapley cost-sharing mechanism for public excludable goods is not Sybil-proof and the limitations of Sybil-proof, truthful, upper semicontinuos and no deficit mechanisms. As reported in \cite{mazorra2023cost}, the creation of Sybils can potentially reduce social welfare in some mechanisms and cause negative externalities. As shown previously, when considering Sybil-proof mechanisms we increase the worst-case social cost from $\mathcal H_n$ to $(n+1)/2$ for the public excludable good problem. However, does this doom the economic efficiency of cost-sharing mechanisms with an unknown number of agents? In this section, we will argue that does not in the case of the Shapley value mechanism for symmetric submodular monotone cost functions. To analyze it, we will need to extend the model's assumptions and accept that agents have private beliefs about the actions of other agents. We will see that if we consider the Sybil-extension of the Shapley value mechanism, we will have a cost-sharing mechanism that is no-deficit, and has welfare bounded by the welfare of cost-sharing mechanism assuming that the number of players is known and cannot generate Sybil identities. We start introducing this property for general mechanisms.

Let $\mathcal M$ be a one-parametric truthful and individual rational anonymous mechanism and $\textbf{Sy}(\mathcal M)$ be its Sybil extension. And let $\mathcal W^{\mathcal M}$ and $\mathcal W^{\textbf{Sy}(\mathcal M)}$ be the social welfare maps, that take as input the agents actions and their true valuations and outputs the social welfare of the outcome. The definition of social welfare strictly depends on the mechanism being studied, in our case we will use the social cost $\pi$ defined in the preliminaries. 
Now, we restrict to mechanisms that its Sybil extension have a weak dominant strategy set $B(v_i)$ for every player $i$ with valuation $v_i$.
\begin{definition} We say that the mechanism $\mathcal M$ is \textit{Sybil welfare invariant} if the welfare under private beliefs with full support of the Sybil extension mechanism is greater or equal than the original mechanism. That is, for every $n\in\mathbb N$ the following inequality holds
\begin{equation*}
\mathcal W^{\mathcal M}(v,v)\leq \mathcal W^{\textbf{Sy}(\mathcal M)}(b,v)
\end{equation*}
where $v=(v_1,...,v_n)\in\mathbb R^n_+$, and $(v,b)\in\mathcal Q = \{(v,b): v\in\mathbb R^n_+,b\in \prod_{i=1}^nB(v_i)\}$.
\end{definition}
In the case of cost-sharing mechanism, our notation of welfare is given by the social cost function $\pi$. Defining $\mathcal W = -\pi$, a cost-sharing mechanism is Sybil welfare invariant if and only if
\begin{equation}
 \mathcal \pi(v,v)\geq \mathcal \pi^{\textbf{Sy}(\mathcal M)}(v,b),
\end{equation}
where $\pi(v,b)$ is the social cost when the agents report $b$ and have true valuations $v$. That is, if the allocation set is $S$ when reporting $b$, $\pi(v,b) = \sum_{i\not\in S} v_i+C(S)$. Similarly, $\pi^{\textbf{Sy}(\mathcal M)}(v,b)$ is defined as the social cost of the Sybil cost-sharing mechanism extensions when the agents report $b$ and their true valuations are $v$. More formally, let $K_i$ be the set of Sybils of agent $i$, $S$ the allocation set (taking into account the Sybils) and $S' = \{i\in[n]: K_i\cap S\not=\emptyset\}$, the Sybil social cost function is defined as $\pi^{\textbf{Sy}(\mathcal M)}(v,b) = \sum_{i\in S'}v_i+C(S')$.

In other words, Sybil welfare invariant mechanisms are those mechanisms such that its Sybil extension have weak dominant strategy sets for every valuation and that all NESPB have at least the same welfare as the output of the mechanism with truthful reports. A Sybil welfare invariant mechanism ensures that the addition of Sybils does not result in a worse outcome in terms of welfare. Essentially, this means the mechanism is resilient to the negative impact of Sybils. Sybil welfare invariant mechanisms are particularly valuable in scenarios where the primary goal is to safeguard the system against loss of welfare due to false-name strategies. They are suitable in environments where Sybil strategies are a concern but eliminating them entirely is not feasible, or identify the identities reported to the mechanism or too costly in terms of ex-post economic efficiency.
An example of Sybil welfare invariant mechanism is a Sybil-proof mechanism. For example, a second price auction with private valuations is Sybil welfare invariant with the welfare map being the maximum valuations among the bidders. In the following, we will see that the Shapley value mechanism is also Sybil welfare invariant even though is not Sybil-proof. Another interpretation of Sybil-welfare invariant mechanisms is the property that hold those mechanisms such that their (subjective) Bayesian Price of anarchy \cite{roughgarden2017price} does not increase when agents can employ Sybil strategies under the no-overbidding condition. More formally, let NESP be the set of subjective Bayes Nash equilibrium holding the no-overbidding condition, the Bayesian price of anarchy of the Sybil extension mechanism is
\begin{equation*}
    \text{BPoA}^{\textbf{Sy}} = \sup_{v,\text{$\sigma$ is NESP}}\frac{\text{Opt}(v)}{\mathbb E_{b\sim \sigma}[\mathcal W^{\textbf{Sy}(\mathcal M)}(b,v)]}
\end{equation*}
where $\text{OPT}(v)$ is the optimal welfare with the valuation profile $v$. Therefore,
\begin{equation}
    \text{BPoA}^{\textbf{Sy}} =  \sup_{v,\text{$\sigma$ is NESP}} \frac{W^{\mathcal M}(v,v)}{\mathbb E_{b\sim \sigma}[\mathcal W^{\textbf{Sy}(\mathcal M)}(b,v)]}\frac{\text{Opt}(v)}{W^{\mathcal M}(v,v)}\leq \sup_v \frac{\text{OPT}(v)}{\mathcal W^{\mathcal M}(v,v)} = \text{PoA}
\end{equation}
deducing the previous claim.

Now we will focus on proving that the Shapley value mechanism with symmetric submodular cost functions is Sybil welfare invariant. To do so, we will break the result in different propositions.
\begin{proposition}\label{prop:shapleyadd} The Shapley value mechanism over symmetric submodular cost functions $C$ holds:
\begin{enumerate}
    \item Strong-monotonicity. In particular, if a player $i$ decides to make a Sybil strategy and commit extra bids, then all the other players' utility will not decrease.
    \item At most $\mathcal H_n-$approximate. 
\end{enumerate}  
\end{proposition}
The proof is mechanical in nature and, for the sake of brevity, is relegated to the appendix for interested readers. We know by the Proposition \ref{prop:cost} that the Shapley value mechanism is, in general, not Sybil-proof. To prove that the Shapley value mechanism with symmetric submodular cost functions is a Sybil welfare mechanism, we will study the optimal Sybil-strategies of the Shapley-value mechanism with more detail. We will compute subweak dominant strategy sets and weak dominant strategy set of the game induced by a given valuation $v_i$ and the  Sybil Shapley value mechanism. We will use it to prove that this mechanism is Sybil welfare invariant.
We will assume that the cost function $C:2^{\mathbb N}\rightarrow \mathbb R_+$ is monotone non-decreasing, symmetric and submodular, and so, there exists $f:\mathbb R_+\rightarrow\mathbb R_+$ such that $f(0)=0$, is monotone non-decreasing, and concave such that $C(S) = f(|S|)$.

In the Sybil-extension mechanism, the space of actions of a player is 
\begin{equation}
    \mathcal A = \{(b_1,...,b_k,...)\mid \forall i\geq1, b_i\geq0,\text{ and } b_i=0\text{ for all but finitely many }i\}.   
\end{equation}
Now, if the private valuation of a player is $v$, we consider the following subset of actions
\begin{equation}
\begin{aligned}
B(v) =\bigcup_{\sigma\in S_\infty}\sigma\cdot\left\{(x_1,...,x_k,...)\in\mathcal A: x_1=v, v/l\geq x_l\text{ for }l\geq2\right\}.
\end{aligned}
\end{equation}
In the following lemma, we will prove that $B(v)$ is a subweak dominant strategy set, and so, rational agents will prefer to choose strategies from $B(v)$.
\begin{proposition} \label{lemma:prev} If an agent has valuation $v$, then the following holds:
\begin{enumerate}
    \item The set of strategies $B(v)$ is a subweak dominant strategy set of $\mathcal A$. That is, for every action $\textbf{b}_i\in \mathcal A\setminus B(v)$, there is a $\textbf{z}_i\in B(v)$ such that
\begin{equation}\label{eq:lemma}
    u_i(\textbf{z},\bb_{-i})\geq u_i(\bb_i,\bb_{-i})\text{ for every tuple of actions } \bb_{-i}. 
\end{equation}
    \item There exists a unique weak dominant strategy set $\overline{B(v)}$. Moreover, for every element $\textbf{b}\in\overline{B(v)}$ there is an element in $\textbf{z}\in B(v)$ such that $u_i(\textbf{z}_i,\bb_{-i})= u_i(\bb_i,\bb_{-i})$  for every tuple of actions  $\bb_{-i}$. 
    \item For any vector of reports $\textbf{b}\in\prod_{i=1}^n \overline{B(v_i)}$, and $\textbf{z}\in\prod_{i=1}^n \overline{B(v_i)}$ such that $\textbf{z}_i=_{u_i}\textbf{b}_i$ for $i=1,...,n$, the social cost of both reports are the same, i.e $\pi(v,\textbf{b}) =\pi(v,\textbf{z})$.
\end{enumerate}
\end{proposition}
The proof follows directly from the previously established lemmas, employing straightforward, mechanical arguments. For clarity and conciseness, it is provided in the appendix.
Now, let's argue why rational agents will choose strategies on the set $\overline{B(v)}$.
Since the Sybil Shapley value mechanism is not truthful, the agents can maximize its utility by being strategic. Suppose now the valuation of the player is $v$ and that the bids $\bb_{-i}$ are drawn from a distribution $\mathcal D$ over $\mathcal A$. Therefore, a rational agent maximizes its utility and, therefore wants to solve the optimization problem
\begin{equation*}
    \underset{x\in\mathcal A}{\text{argmax}}\quad\mathbb E_{\bb_{-i}\sim\mathcal D}[u_i(x,\bb_{-i})].
\end{equation*}
Now, let $x^\star$ be an element that maximizes the expected utility, then we know that there exists $z\in B(v)$ such that $u_i(z,\bb_{-i})\geq u_i(x^\star,\bb_{-i})$ for all $\bb_{-i}\in\mathcal A$, and the inequality is strict for some element in $\bb_{-i}\in\mathcal A$. In particular, we deduce that $\mathbb E_{\bb_{-i}\sim\mathcal D}[u_i(z,\bb_{-i})]\geq \mathbb E_{\bb_{-i}\sim\mathcal D}[u_i(x^\star,\bb_{-i})]$, and so 
\begin{equation*}
    \underset{x\in\mathcal A}{\text{argmax}}\quad\mathbb E_{\bb_{-i}\sim\mathcal D}[u_i(x,\bb_{-i})]\cap  \underset{x\in\overline{B(v)}}{\text{argmax}}\quad\mathbb E_{\bb_{-i}\sim\mathcal D}[u_i(x,\bb_{-i})]\not=\emptyset.
\end{equation*}
Moreover, if the distribution $\mathcal D$ has full support 
(all open sets under the final topology have non-zero probability), it holds
\begin{equation}
 \underset{x\in\mathcal A}{\text{argmax}}\quad\mathbb E_{\bb_{-i}\sim\mathcal D}[u_i(x,\bb_{-i})]=  \underset{x\in\overline{B(v)}}{\text{argmax}}\quad\mathbb E_{\bb_{-i}\sim\mathcal D}[u_i(x,\bb_{-i})]
\end{equation}
by using that $\overline{B(v)}$ is a weak dominant strategy set.
So, we will assume that, in equilibrium, agents will choose strategies from $\overline{B(v)}$. Now, since the Shapley value mechanism is no longer truthful in the Sybil extension, the agents will choose actions or strategies that maximize their expected utility over their private beliefs. However, as we will see, this will not have an impact on the worst-case social cost of the Shapley value mechanism.

\begin{theorem} The Shapley value mechanism for public excludable goods with symmetric submodular monotone cost function is Sybil welfare invariant.
\end{theorem}
\begin{proof} First, recall that since the cost function is symmetric and submodular, there exists a monotone non-decreasing concave function $f$ such that $C(S) = f(|S|)$ for every $S\subseteq \mathbb N$. Since $f$ is concave, we have that $f(m)/m\leq f(n)/n$ for all $n\leq m$. Also, we deduce that
$(m-n)f(m)/m\geq f(m)-f(n)$ for all $n\leq m$.
Now, suppose that there are $n$ players with private valuations $v_1,...,v_n$ and private beliefs with full support $\mathcal D_1,...,\mathcal D_n$. By Lemma \ref{lemma:strict}, the agents take actions in $\prod_{i=1}^n\overline{B(v_i)}$. 
Now, let $x=(x_1,...,x_n)$ be the actions taken by the players. First, we will see that wlog we can assume that $(x_1,...,x_n)\in\prod_{i=1}^n B(v_i)$ by the point 3 of Proposition \ref{lemma:prev}.
Let $S(v)$ (resp. $S(x)$) be the set of players that have some identity having access to the public good when reporting $v$ (resp. $x$). By definition of $B(v)$, the first element is $v$, therefore $S(v)\subseteq S(x)$, and so $\sum_{i\in S(x)} v_i\geq \sum_{i\in S(v)}v_i$. Since the mechanism is budget-balanced, the sum of payments is $C(S)$ with winners set $S$. We will prove that $\pi(S(x))\leq \pi(S(v))$ by cases.

\textbf{Case 1} $S(x)=\emptyset$. We know that $S(v)\subseteq S(x)$ and so $S(v)=\emptyset$, in both cases we have the null allocation and so the social cost of both cases coincide. Therefore $\pi(S(v))= \pi(S(x))\leq \mathcal H_n\pi(S^\star)$.

\textbf{Case 2} $S(x)\not=\emptyset$.  For every $i\in S(x)$, let $k_i$ be the number of $i$ Sybil identities $k_i$ such that the public good is allocated. Since $x\in B(v_i)$, $v_i/j \geq x^{j}_i$ for $j=1,...,k_i$, in particular $v_i\geq f(|S(x)|)/|S(x)|$.  Therefore, 
\begin{equation*}
    \sum_{i\in S(x)\setminus S(v)}v_i\geq(|S(x)|-|S(v)|)f(|S(x)|)/|S(x)|\geq f(|S(x)|)-f(|S(v)|)
\end{equation*}
where the last inequality is deduced from $f$ being concave and $|S(v)|\leq |S(x)|$. So,
\begin{align*}
    \pi(S(x)) &= f(|S(x)|) + \sum_{i\not\in S(x)}v_i=  f(|S(x)|) + \sum_{i\not\in S(v)}v_i - \sum_{i\in S(x)\setminus S(v)}v_i\\
              &\leq  f(|S(x)|) +  \sum_{i\not\in  S(v)}v_i - f(|S(x)|)+f(|S(v)|) \\
              &= \pi(S(v)).
\end{align*}
\end{proof}
As a corollary, it can be shown that if agents do not overbid, the Bayesian price of anarchy of the Sybil extension of the Shapley value mechanism is $\mathcal H_n$. In summary, we have proved that even when the number of agents is unknown to both the mechanism designer, and the agents participating in the mechanism, the Shapley value mechanism for submodular monotone cost functions has the same worst-case social cost as the same mechanism with known number of agents. Therefore, we have shown the robustness of Shapley value mechanism over permissionless environments such as Peer-to-Peer (P2P) Networks and decentralized finance (DeFi) platforms, and in particular can be practical and robust when members of a decentralized autonomous organizations (DAOs) want to deploy a public excludable good. Since the Shapley value mechanism is no longer truthful, to implement this mechanism in a public blockchain will be necessary to make some small adjustments. To maintain the same worst-case equilibria under false-name strategies the mechanism will have to shield the bids, similar to \cite{ferreira2020credible}, by using different cryptographic tools such as commit-and-reveal or multiparty computation protocols. For example, the mechanism could have two phases. The commit phase and the reveal phase. In the commit-phase, agents send a commitment of a bid and some cash as a collateral (for example the cost of the public good). After the finish of the first phase, the agents reveal their bid, and the allocation and the payment is computed following the rules of the Shapley value mechanism. If an agent does not reveal its bid in time, they lose their collateral, making weak dominant strategy to reveal their bid. Without this two-phase mechanism or a trusted third party that keeps the bids private, some agents would have access to other agents' bids, changing the structure of the mechanism from one-shoot mechanisms to sequential mechanisms with multiple rounds. 
\section{Conclusion}\label{section:conclusions}
In this paper, we have formalized false-name strategies in cost-sharing mechanisms. We established an impossibility result, indicating that many mechanisms from existing literature are vulnerable to these strategies. Furthermore, we characterized the worst-case welfare for mechanisms that satisfy the properties of individual rationality, symmetry, truthfulness, upper semicontinuous and Sybil-proofness. These mechanisms have a worst-case welfare of at least $\Omega(n)$, and we demonstrated that this bound is asymptotically tight. Additionally, we introduced the concept of the Sybil welfare invariant property and showed that the Shapley value mechanism possesses this property. This means that regardless of the priors held by agents, the Shapley value mechanism with sybils achieves the same worst-case welfare as the Shapley value mechanism without sybils. 
As a direction for future research, we aim to explore the vulnerabilities of combinatorial cost-sharing mechanisms to false-name strategies and determine whether mechanisms cited in the literature, such as \cite{dobzinski2017combinatorial} and \cite{birmpas2022cost}, are Sybil welfare invariant. To do so, we will need to extend the definition of Sybil welfare invariant under combinatorial domains. Also, we leave as future work, to see if Theorem \ref{theorem:lower} holds for weaker conditions such as removing the deterministic condition. Finally, in the cost-sharing literature, we aim to study Sybil-proof and Sybil-welfare invariant mechanisms with general cost functions $C$.
Beyond the cost-sharing literature, we aim to study if mechanisms such as the one proposed in \cite{bahrani2023bidders} are Sybil welfare invariant and if not, make adjustments to the mechanism to maintain the worst-case welfare under private beliefs.

\section{Acknowledgments}
The author would like to express sincere gratitude to the Ethereum Foundation for generously funding the scholarship that made this research possible. Special thanks are also extended to Vanesa Daza, for her invaluable feedback and meticulous corrections. 

\printbibliography

@String{JACM = "J. ACM" }

@String{Computer = "{IEEE} Computer" }

@String{Academic = "Academic Press" }

@String{Springer = "Springer-Verlag" }

@book{krishna2009auction,
  title={Auction theory},
  author={Krishna, Vijay},
  year={2009},
  publisher={Academic press}
}

@article{yokoo2004effect,
  title={The effect of false-name bids in combinatorial auctions: New fraud in Internet auctions},
  author={Yokoo, Makoto and Sakurai, Yuko and Matsubara, Shigeo},
  journal={Games and Economic Behavior},
  volume={46},
  number={1},
  pages={174--188},
  year={2004},
  publisher={Elsevier}
}

@inproceedings{douceur2002sybil,
  title={The sybil attack},
  author={Douceur, John R},
  booktitle={International workshop on peer-to-peer systems},
  pages={251--260},
  year={2002},
  organization={Springer}
}

@inproceedings{dinger2006defending,
  title={Defending the sybil attack in p2p networks: Taxonomy, challenges, and a proposal for self-registration},
  author={Dinger, Jochen and Hartenstein, Hannes},
  booktitle={First International Conference on Availability, Reliability and Security (ARES'06)},
  pages={8--pp},
  year={2006},
  organization={IEEE}
}

@inproceedings{yu2006sybilguard,
  title={Sybilguard: defending against sybil attacks via social networks},
  author={Yu, Haifeng and Kaminsky, Michael and Gibbons, Phillip B and Flaxman, Abraham},
  booktitle={Proceedings of the 2006 conference on Applications, technologies, architectures, and protocols for computer communications},
  pages={267--278},
  year={2006}
}

@inproceedings{yu2008sybillimit,
  title={Sybillimit: A near-optimal social network defense against sybil attacks},
  author={Yu, Haifeng and Gibbons, Phillip B and Kaminsky, Michael and Xiao, Feng},
  booktitle={2008 IEEE Symposium on Security and Privacy (sp 2008)},
  pages={3--17},
  year={2008},
  organization={IEEE}
}

@inproceedings{so2011defending,
  title={Defending against sybil nodes in bittorrent},
  author={So, Jung Ki and Reeves, Douglas S},
  booktitle={International Conference on Research in Networking},
  pages={25--39},
  year={2011},
  organization={Springer}
}

@article{zhang2019double,
  title={Double-spending with a sybil attack in the bitcoin decentralized network},
  author={Zhang, Shijie and Lee, Jong-Hyouk},
  journal={IEEE transactions on Industrial Informatics},
  volume={15},
  number={10},
  pages={5715--5722},
  year={2019},
  publisher={IEEE}
}

@inproceedings{mazorra2022price,
  title={Price of MEV: Towards a Game Theoretical Approach to MEV},
  author={Mazorra, Bruno and Reynolds, Michael and Daza, Vanesa},
  booktitle={Proceedings of the 2022 ACM CCS Workshop on Decentralized Finance and Security},
  pages={15--22},
  year={2022}
}

@inproceedings{todo2011false,
  title={False-name-proof mechanism design without money},
  author={Todo, Taiki and Iwasaki, Atsushi and Yokoo, Makoto},
  booktitle={The 10th International Conference on Autonomous Agents and Multiagent Systems-Volume 2},
  pages={651--658},
  year={2011}
}

@inproceedings{wagman2008optimal,
  title={Optimal False-Name-Proof Voting Rules with Costly Voting.},
  author={Wagman, Liad and Conitzer, Vincent},
  booktitle={AAAI},
  volume={8},
  pages={190--195},
  year={2008}
}

@inproceedings{iwasaki2010worst,
  title={Worst-case efficiency ratio in false-name-proof combinatorial auction mechanisms},
  author={Iwasaki, Atsushi and Conitzer, Vincent and Omori, Yoshifusa and Sakurai, Yuko and Todo, Taiki and Guo, Mingyu and Yokoo, Makoto},
  booktitle={Proceedings of the 9th International Conference on Autonomous Agents and Multiagent Systems: volume 1-Volume 1},
  pages={633--640},
  year={2010}
}

@article{fioravanti2022false,
  title={False-name-proof and strategy-proof voting rules under separable preferences},
  author={Fioravanti, Federico and Mass{\'o}, Jordi},
  journal={Available at SSRN 4175113},
  year={2022}
}

@article{yokoo2001robust,
  title={Robust combinatorial auction protocol against false-name bids},
  author={Yokoo, Makoto and Sakurai, Yuko and Matsubara, Shigeo},
  journal={Artificial Intelligence},
  volume={130},
  number={2},
  pages={167--181},
  year={2001},
  publisher={Elsevier}
}

@article{sher2012optimal,
  title={Optimal shill bidding in the VCG mechanism},
  author={Sher, Itai},
  journal={Economic Theory},
  volume={50},
  pages={341--387},
  year={2012},
  publisher={Springer}
}

@inproceedings{dobzinski2008shapley,
  title={Is Shapley cost sharing optimal?},
  author={Dobzinski, Shahar and Mehta, Aranyak and Roughgarden, Tim and Sundararajan, Mukund},
  booktitle={Algorithmic Game Theory: First International Symposium, SAGT 2008, Paderborn, Germany, April 30-May 2, 2008. Proceedings 1},
  pages={327--336},
  year={2008},
  organization={Springer}
}

@article{mazorra2023cost,
  title={The Cost of Sybils, Credible Commitments, and False-Name Proof Mechanisms},
  author={Mazorra, Bruno and Della Penna, Nicol{\'a}s},
  journal={arXiv preprint arXiv:2301.12813},
  year={2023}
}

@article{moulin1999incremental,
  title={Incremental cost sharing: Characterization by coalition strategy-proofness},
  author={Moulin, Herv{\'e}},
  journal={Social Choice and Welfare},
  volume={16},
  pages={279--320},
  year={1999},
  publisher={Springer}
}

@article{myerson1981optimal,
  title={Optimal auction design},
  author={Myerson, Roger B},
  journal={Mathematics of operations research},
  volume={6},
  number={1},
  pages={58--73},
  year={1981},
  publisher={INFORMS}
}

@book{sundararajan2009trade,
  title={Trade-offs in cost sharing},
  author={Sundararajan, Mukund},
  year={2009},
  publisher={Stanford University}
}

@inproceedings{dobzinski2017combinatorial,
  title={Combinatorial cost sharing},
  author={Dobzinski, Shahar and Ovadia, Shahar},
  booktitle={Proceedings of the 2017 ACM Conference on Economics and Computation},
  pages={387--404},
  year={2017}
}

@article{harsanyi1968part,
  title={Part II: Bayesian equilibrium points},
  author={Harsanyi, Jhon C},
  journal={Management Science},
  volume={14},
  pages={320--334},
  year={1968}
}

@article{varian2009online,
  title={Online ad auctions},
  author={Varian, Hal R},
  journal={American Economic Review},
  volume={99},
  number={2},
  pages={430--434},
  year={2009},
  publisher={American Economic Association}
}

@article{birmpas2022cost,
  title={Cost sharing over combinatorial domains},
  author={Birmpas, Georgios and Markakis, Evangelos and Sch{\"a}fer, Guido},
  journal={ACM Transactions on Economics and Computation},
  volume={10},
  number={1},
  pages={1--26},
  year={2022},
  publisher={ACM New York, NY}
}

@article{maskin1985auction,
  title={Auction theory with private values},
  author={Maskin, Eric S and Riley, John G},
  journal={The American Economic Review},
  volume={75},
  number={2},
  pages={150--155},
  year={1985},
  publisher={JSTOR}
}

@article{moulin2001strategyproof,
  title={Strategyproof sharing of submodular costs: budget balance versus efficiency},
  author={Moulin, Herv{\'e} and Shenker, Scott},
  journal={Economic Theory},
  volume={18},
  pages={511--533},
  year={2001},
  publisher={Springer}
}

@article{bahrani2023bidders,
  title={When Bidders Are DAOs},
  author={Bahrani, Maryam and Garimidi, Pranav and Roughgarden, Tim},
  journal={arXiv preprint arXiv:2306.17099},
  year={2023}
}

@inproceedings{bleischwitz2008group,
  title={Group-strategyproof cost sharing for metric fault tolerant facility location},
  author={Bleischwitz, Yvonne and Schoppmann, Florian},
  booktitle={International Symposium on Algorithmic Game Theory},
  pages={350--361},
  year={2008},
  organization={Springer}
}

@inproceedings{brenner2007cost,
  title={Cost sharing methods for makespan and completion time scheduling},
  author={Brenner, Janina and Sch{\"a}fer, Guido},
  booktitle={Annual Symposium on Theoretical Aspects of Computer Science},
  pages={670--681},
  year={2007},
  organization={Springer}
}

@inproceedings{chawla2006optimal,
  title={Optimal cost-sharing mechanisms for steiner forest problems},
  author={Chawla, Shuchi and Roughgarden, Tim and Sundararajan, Mukund},
  booktitle={Internet and Network Economics: Second International Workshop, WINE 2006, Patras, Greece, December 15-17, 2006. Proceedings 2},
  pages={112--123},
  year={2006},
  organization={Springer}
}

@article{gupta2015efficient,
  title={Efficient cost-sharing mechanisms for prize-collecting problems},
  author={Gupta, Anupam and K{\"o}nemann, Jochen and Leonardi, Stefano and Ravi, Ramamoorthi and Sch{\"a}fer, Guido},
  journal={Mathematical Programming},
  volume={152},
  pages={147--188},
  year={2015},
  publisher={Springer}
}

@article{feigenbaum2003hardness,
  title={Hardness results for multicast cost sharing},
  author={Feigenbaum, Joan and Krishnamurthy, Arvind and Sami, Rahul and Shenker, Scott},
  journal={Theoretical Computer Science},
  volume={304},
  number={1-3},
  pages={215--236},
  year={2003},
  publisher={Elsevier}
}

@article{roughgarden2009quantifying,
  title={Quantifying inefficiency in cost-sharing mechanisms},
  author={Roughgarden, Tim and Sundararajan, Mukund},
  journal={Journal of the ACM (JACM)},
  volume={56},
  number={4},
  pages={1--33},
  year={2009},
  publisher={ACM New York, NY, USA}
}

@inproceedings{ferreira2020credible,
  title={Credible, truthful, and two-round (optimal) auctions via cryptographic commitments},
  author={Ferreira, Matheus VX and Weinberg, S Matthew},
  booktitle={Proceedings of the 21st ACM Conference on Economics and Computation},
  pages={683--712},
  year={2020}
}

@book{myerson1989mechanism,
  title={Mechanism design},
  author={Myerson, Roger B},
  year={1989},
  publisher={Springer}
}

@article{dufwenberg2002existence,
  title={Existence and uniqueness of maximal reductions under iterated strict dominance},
  author={Dufwenberg, Martin and Stegeman, Mark},
  journal={Econometrica},
  volume={70},
  number={5},
  pages={2007--2023},
  year={2002},
  publisher={Wiley Online Library}
}

@article{wiszniewska2016belief,
  title={Belief distorted Nash equilibria: introduction of a new kind of equilibrium in dynamic games with distorted information},
  author={Wiszniewska-Matyszkiel, Agnieszka},
  journal={Annals of Operations Research},
  volume={243},
  pages={147--177},
  year={2016},
  publisher={Springer}
}

@article{roughgarden2017price,
  title={The price of anarchy in auctions},
  author={Roughgarden, Tim and Syrgkanis, Vasilis and Tardos, Eva},
  journal={Journal of Artificial Intelligence Research},
  volume={59},
  pages={59--101},
  year={2017}
}
\appendix
\section{Appendix}

\subsection{Notation}
\vspace{0.5cm}
\begin{tabular}{c m{12cm} }
\hline
Symbol & Description \\
\hline
\( [n] \) & Set \( \{1, \ldots, n\} \). \\
$2^{[n]}$ & Set of subsets of $[n]$.\\
$\mathbb N$ & Set of natural numbers $\{1,2,3,4,...\}$.\\
\( S_\infty \) & Set of permutations of $\mathbb N$. \\
\( u_i \) & Player \( i \)'s utility function. \\
\( v_i \) & Player \( i \)'s private valuation. \\
\( \x \) & Allocation map. \\
\( \p \) & Payment map. \\
$x_i$ & The $i$-th component of the allocation map $\textbf{x}$.\\
$p_i$ & The $i$-th component of the payment map $\textbf{x}$.\\
$\mathcal{H}_n$ & $n$-th harmonic number. \\
$\mathbb{R}^\infty $ & Space of sequences of real numbers where only finitely many terms are non-zero. \\
$\mathbb{R}^\infty_+$ & Space of sequences of non-negative real numbers where only finitely many terms are non-zero. \\
\hline
\end{tabular}
\subsection{Proofs}

\textbf{Proof }\ref{lemma:strict} Let's first prove it assuming that the subweak dominant strategy set $B_i(t_i)$ is sequentially compact. Let's define the following relation. We say that $x\geq_{u_i}y$ if and only if $u_i(t_i,x,z)\geq u_i(t_i,y,z)$ for every $z\in A_{-i}$. Observe that the relation is reflexive and transitive, but is not symmetric. We say that $x=_{u_i} y$ if and only if $x\geq_{u_i}y$ and $y\geq_{u_i}x$ and $x>_{u_i} y$ if $x\geq_{u_i} y$ and $x\not=_{u_i} y$. For every, $x\in A_i$, we define $[x]_{\geq_{u_i}} = \{y\in A_i: x\geq_{u_i}y\}$. First, we consider the subset $\overline{B_i(t_i)}=\{x\in A_i:\nexists y\in B_i(t_i)\text{ s.t. }y>_{u_i}x_i\}$. 

We claim that $\overline{B_i(t_i)}$ is non-empty and is a weak dominant strategy set. First, lets prove that is non-empty. We will use Zorn's lemma and that $B_i(t_i)$ is sequentially compact. Consider a chain in $B_i(t_i)$, that is a totally order sequence $\{x_n\}_{n\in\mathbb N}$, by sequentially compactness of $B_i(t_i)$, there exists a subsequence $x_{n_j}$ that converges to an element $x^\star\in B_i(t_i)$. Using the upper-continuity of $u_i$ follows easily that $x^\star\geq_{u_i} x_n$ for all $n\in\mathbb N$. 
Therefore, every chain has a maximal element. Now, by Zorn's lemma, there exists at least one maximal element $x$ in $B_i(t_i)$. Now, suppose that $\overline{B_i(t_i)}$ is an empty set then, there is no maximal element in $A_i$, implying that for every $x\in A_i$ there is an element $y\in A_i$ such that $y>_{u_i}x$. Let $x^\star$ be a maximal element of $B_i(t_i)$. 
Then, there exists $y\in A_i$ such that $y>_{u_i}x^\star$. Since $B_i(t_i)$ is subweak dominant strategy set, there exists $y'$ such that $y'>_{u_i}x^\star$, but this contradicts the fact that $x^\star$ is a maximal element in $B_i(t_i)$, therefore $\overline{B_i(t_i)}$ is non-empty. Moreover, with this argument, we have seen that all maximal elements of $B_i(t_i)$ are elements of $\overline{B_i(t_i)}$. 

Now, we claim that $\overline{B_i(t_i)}$ is the unique strictly dominant strategy set. Let's prove it. Given an element $a_i\in A_i\setminus B_i(t_i)$ there is an element $a_i'\in B_i(t_i)$ such that $a_i'\geq_{u_i} a_i$. Now, since $\overline{B_i(t_i)}$ contain all maximal elements, we have that there exists an element $a_i''\geq_{u_i}a_i'$, deducing the first property of weak dominant strategy sets. The second condition clearly holds by definition.

The uniqueness follows similarly. If there are two different strictly dominant sets $M_1$ and $M_2$, then wlog there exists $x\in M_2\setminus M_1$. Since $M_1$ is a strictly dominant strategy set, there is $y\in M_1$ such that $y>_{u_i} x$. But this contradicts the fact that $M_1$ is a dominant strategy set.

Now let's suppose that there exists a chain $C_1\subseteq C_2\subseteq ...$ that holds the hypothesis of the lemma. For every $j\in\mathbb N$,  $B_i(t_i)\cap C_j$ is a dominant strategy set when restricting the game to $C_j$. And so by the previous argument, there exists $\overline{B}(t_i)^j_i$ strictly dominant strategy set of the normal form game where the agent $i$ has action space $C_i$. Now, the set $\cup_{j\in\mathbb N} \overline{B}_i(t_i)^j$ is a strictly dominant strategy set in $A_i$.

\textbf{Proof }\ref{prop:shapleyadd}
\begin{enumerate}
    \item The Shapley value is the cost-sharing mechanism with cost-sharing method $\zeta(S) = 1/|S|$. Now let $\bb$ and $\bb'$ be two bid vector profiles such that $\bb\geq \bb'$. 
    Let $S'$ be the allocation set of $\bb'$ and $S$ be the allocation set of $\bb$. Then, for every $i\in S'$, we have that $b'_i\geq C(S')/|S'|$. On the other hand, $b_i\geq b_i'\geq C(S')/|S'|$. Therefore, $S'\subseteq \text{argmax}_{X}\{|X|:b_i\geq C(X)/|X|\}=S$. The utility of other players will not decrease under more bids, since the allocation set will be at least $S$, and the payment will be at most $C(S)/|S|$.
    \item To do so, we will see that the allocation of the Shapley value mechanism $S_1$ with subadditive symmetric valuations contains the set of agents allocated $S_2$ allocated using the Hybrid mechanism provided in \cite{dobzinski2008shapley}. In this proposition, we use that the cost function is submodular, otherwise, the Shapley value mechanism would not be truthful in general \cite{moulin2001strategyproof}. The hybrid mechanism consists of the following:
    \begin{mybox2}{Hybrid mechanism}
        \begin{enumerate}
        \item Accept bids $b_1,...,b_n$.
        \item Take $S^\star\in\text{argmax}_S\{\sum_{i\in S} b_i-C(S)\}$.
        \item Initialize $S:=S^\star$.
        \item If $b_i\geq C(S^\star)/|S|$ for every $i\in S$, then halt with winners $S$.
        \item Let $i^\star\in S$ be a player with $b_{i^\star}<C(S^\star)/|S|$.
        \item Set $S\leftarrow S\setminus\{i\}$ and return to Step 4.
        \item Charge each winner $i\in S$ a payment equal to the minimum bid at which $i$ would continue to win (holding $b_{-i}$ fixed).
    \end{enumerate}
    \end{mybox2}
    First, wlog, we order the bids $b_1\geq...\geq b_n$. Observe that since the cost function is symmetric and subadditive, if $k\in S_i$ for $i=1,2$, then $[k]\subseteq S_i$ for $i=1,2$. Let $k_i$ be the largest element of $S_i$. So, proving that $S_2\subseteq S_1$ is equivalent to prove that $k_2\leq k_1$. Then, $b_{l}\geq C(S^\star)/|S_2|\geq C(S_2)/|S_2|$ for $l=1,...,k_2$, where $S^\star$ is the set that maximizes  $\sum_{i\in S}b_i - C(S)$. In particular, $k_2\in\{i:b_i\geq C(S)/|S|\}$ and so $k_2\leq \text{argmax}\{b_i\geq C([i])/i\}$. $\square$
\end{enumerate}
\textbf{Proof }\ref{lemma:prev} With out loss of generality, we will assume that $f(1)=1$.
\begin{enumerate}
    \item Given a vector $\bb_i = (b_1,...,b_k,0,...)$ (wlog we assume $b_1\geq b_2\geq...\geq b_k$), we consider the vector $\textbf{z} = (z_1,...,z_k,0,...)$ defined by
\begin{align*}
    z_1 &= v\\
    z_{l} &= \min\left\{ b_l,\frac{v}{l},\frac{1}{l}\right\},\text{ for }l=2,...,k.
\end{align*}
Clearly, $z_1\geq...\geq z_l$.  We will prove that $\textbf{z}$ holds the inequality \ref{eq:lemma} by cases will depend on the number of Sybils that both bid vector profiles will have access to the public good. Let $j(\bb_i)$ (resp. $j(\textbf{z})$) be the total number of Sybil identities of agent $i$ that are allocated when reporting $\bb_i$ (resp. $\textbf{z}$). Similarly, let $n(\bb_i)$ (resp. $n(\textbf{z})$) be the total number of identities that are allocated when reporting $\bb_i$ (resp. $\textbf{z}$).

\textbf{Case 1} $j(\textbf{b}_i)=0$ and $j(\textbf{z})\geq1$. 
Suppose that no identity is allocated when reporting $\bb_i$, then, in this case, the utility is zero. Now, when reporting $\textbf{z}$, at most the first identity will have access to the public good, otherwise $z_2\geq f(n(\textbf{z}))/n(\textbf{z})$. As $b_1\geq b_2\geq z_2$, then the first two identities would also be allocated when reporting $\bb_i$, leading to a contradiction. Since the mechanism is incentive compatible, the utility reporting $\textbf{z}$ is greater than zero, proving this case. 

\textbf{Case 2} $j(\textbf{b}_i)\geq 1$, $j(\textbf{z})=0$ and $v\leq 1$. Then, we have the last sybil identity that has access holds $b_j\geq f(n(\textbf{b}_i))/n(\textbf{b}_i)$. On the other hand, since no Sybil identity has access when reporting $\textbf{z}$, it holds that $z_{j(\textbf{b}_i)}<f(n(\textbf{b}_i))/n(\textbf{b}_i)$ and so $v/j(\textbf{b}_i)<f(n(\textbf{b}_i))/n(\textbf{b}_i)$. Since the payment when reporting $\bb_i$ is $j(\textbf{b}_i)f(n(\textbf{b}_i))/n(\textbf{b}_i)$, the utility when reporting $\bb_i$ is negative, and is zero when reporting $\textbf{z}$.

\textbf{Case 3} $j(\textbf{b}_i)\geq 1$, $j(\textbf{z})=0$ and $v>1$. It is not possible since the first bid reported by the agent $i$ is $v>1$ and so that identity has access to the public good, since $v>f(1)/1\geq f(n)/n$ for all $n\in\mathbb N$, leading to $j(\textbf{z})\geq1$.

\textbf{Case 4} $j(\textbf{b}_i)\geq 1$ and $j(\textbf{z})\geq1$. In this case, wlog, we can assume that $b_1\geq v$ and so $\bb_i\geq\textbf{z}$. By construction and the strong-monotonicity of the Shapley value mechanism, if a Sybil identity is allocated when reporting $\textbf{z}$, then the same identity is allocated when the report is $\bb_i$. Therefore, $j(\textbf{b}_i)\geq j(\textbf{z})$ and also $n(\textbf{b}_i)\geq n(\textbf{z})$. If $j(\textbf{b}_i)= j(\textbf{z})$, then the utility of both cases is the same. So, lets assume that $j(\textbf{b}_i)> j(\textbf{z})$. This implies that $z_{j(\textbf{b}_i)}<f(n(\bb_i))/n(\bb_i)$, and so $v/j(\textbf{b}_i)<f(n(\bb_i))/n(\bb_i)$ (or $1/j(\textbf{b}_i)<f(n(\bb_i))/n(\bb_i)$ if $v\geq1)$). On the other hand, $z_{j(\textbf{z})}\geq f(n(\textbf{z}))/n(\textbf{z})$, and so $v/j(\textbf{z})\geq f(n(\textbf{z}))/n(\textbf{z})$ (or $1/j(\textbf{z})\geq f(n(\textbf{z}))/n(\textbf{z})$ in case $v\geq1$) . Using both equations, in both cases, we deduce that $j(\textbf{b}_i)f(n(\bb_i))/n(\bb_i)>j(\textbf{z})f(n(\textbf{z}))/n(\textbf{z})$. Now, the utility is $u_i(\textbf{b}_i,\bb_{-i}) = v-j(\textbf{b}_i)f(n(\bb_i))/n(\bb_i)$ and $u_i(\textbf{z},\bb_{-i}) = v-j(\textbf{z})f(n(\textbf{z}))/n(z)$ since both have access to the public good and the payments for each Sybil is $f(n(\bb_i))/n(\textbf{b}_i)$ (resp. $f(n(\textbf{z}))/n(\textbf{z})$). And so, we deduce the result.
\item Observe that the space $\mathcal A$ has structure of a metric space with the metric $d(x,y) = \sqrt{\sum_{i=1}^\infty |x_i-y_i|^2}$. To prove the proposition, first we prove the following claim: there is a chain of sets $C_i$ such that $\cup_{i\in\mathbb N} C_i = \mathcal A_i$ and $B(v)\cap C_i$ is sequentially compact for every $v\in\mathbb R_+$. Take the set $C_i = \mathbb R_+^i\times\{0\}^\mathbb N\subseteq \bigoplus_{i=j}^\infty\mathbb R_+$. Clearly the set $B(v)\cap C_i = \bigoplus_{l=1}^i [0,v/l]\times\{0\}^\mathbb N$ is compact since is the product of compact spaces. Therefore, by lemma \ref{lemma:strict}, there exists a strictly dominant strategy set $\overline{B(v)}$.
\item Given bid vector profile $\textbf{b}\in\prod_{i=1}^n\overline{B(v_i)}$ with components $\textbf{b}_i$, consider the element $\textbf{z}_i$ as defined in $1)$. Since $\overline{B(v_i)}$ are strictly dominant strategy sets and $\textbf{z}_i\geq_{u_i} \textbf{b}_i$, we have that $\textbf{z}_i=_{u_i} \textbf{b}_i$. We will see that the set of allocated sybils is the same for the bid vector profiles $\textbf{b}$ and $\textbf{z}$. We know that $\textbf{z}_i =_{u_i}\textbf{b}_i$, and we claim that it implies that for every $j\in\mathbb N$, there exists $n_j$ such that $f(n_j)/n_j\leq\textbf{z}_{ij},\textbf{b}_{ij}< f(n_j-1)/(n_j-1)$. Suppose not, take the minimum $j$ such that the claim does not hold and so, there exists $n\in\mathbb N$ such that $f(n)/n\leq\textbf{z}_{ij}< f(n-1)/(n-1)\leq\textbf{b}_{ij}$. Since $\textbf{z}_{i1}=\textbf{b}_{i1}=v$, we have that $j\geq2$. Consider now the bid vector profile $\textbf{b}_{-i} = \underbrace{(f(n-1)/(n-1),...,f(n-1)/(n-1))}_{n-j-1 \text{ components}}$. By definition of $\textbf{z}_i$ and the previous inequality, we have that $\textbf{z}_{ij}=\min\{v/j,1/j\}$ and so, we deduce that $jf(n-1)/(n-1)>v\geq j f(n)/n$ if $v\leq 1$ and $jf(n-1)/(n-1)>1\geq  jf(n)/n$ otherwise. Now, let's compute the utility of the bid vector profiles $(\textbf{z}_i,\textbf{b}_{-i})$ and $(\textbf{b}_i,\textbf{b}_{-i})$. The allocation of the vector $(\textbf{z}_i,\textbf{b}_{-i})$ if $v<1$ is null, and if $v\geq1$ is just the Sybils of player $i$ inducing in the first case utility of $0$ and in the second case has utility of $v-f(1)$. 

When reporting $\textbf{b}_i$ all elements $\textbf{b}_{il}$ for $1\leq l\leq k$ are greater than $f(n-1)/(n-1)$, therefore, all the identities reported in the bid vector $\textbf{b}_{-i}$ are allocated. Therefore, his payment is at least $j\frac{f(n-1)}{n-1}$ and so, his utility is, at most, $v-j\frac{f(n-1)}{n-1}$. When  $v\leq 1$, we have that $v-j\frac{f(n-1)}{n-1}<0$ since $jf(n-1)/(n-1)>v$ and for $v>1$, we have that $v-f(1)>v-j\frac{f(n-1)}{n-1}$ since $jf(n-1)/(n-1)>1=f(1)$. This contradicts the claim that $\textbf{z}_{u_i}=\textbf{b}_i$, and so for every $j\in\mathbb N$, there exists $n_j$ such that $f(n_j)/n_j\leq\textbf{z}_{ij},\textbf{b}_{ij}< f(n_j-1)/(n_j-1)$. Now, clearly the allocation of both bid vector profiles is the same since the Shapley value mechanism consists of choosing the biggest set $S$ such that all bids are greater or equal $f(|S|)/|S|$.
\end{enumerate}

\end{document}